\documentclass[12pt]{article}

\usepackage{amsthm,amsmath,mathrsfs,amsfonts,amssymb}
\usepackage{graphicx,psfrag,epsf,fixmath,caption,subcaption}
\usepackage{enumerate}
\usepackage{natbib}
\usepackage{nicefrac} \RequirePackage[colorlinks,citecolor=blue,urlcolor=blue]{hyperref}
\usepackage{url} %

\usepackage[multiple]{footmisc} \usepackage[capposition=top]{floatrow}

\newcommand{\blind}{1}
\newtheorem{assumption}{Assumption}

\newtheorem{proposition}{Proposition}

\newtheorem{lemma}{Lemma}

\theoremstyle{definition}

\begin{document}

\def\spacingset#1{\renewcommand{\baselinestretch}{#1}\small\normalsize} \spacingset{1}

\if1\blind
{
  \title{\bf \vspace{-1cm}{\large INFERENCE ON A DISTRIBUTION FROM NOISY DRAWS}}
  \author{ Koen Jochmans\thanks{Address: Toulouse School of Economics, 1 esplanade de l'Universit\'e, 31080 Toulouse, France. E-mail: \texttt{koen.jochmans@tse-fr.eu}.}\\ Toulouse School of Economics \\  Universit\'e  Toulouse 1 Capitole  \and
  \setcounter{footnote}{2}
   Martin Weidner\thanks{Address: University of Oxford, Department of Economics, Manor Road, Oxford OX1 3UQ, United Kingdom. E-mail: \texttt{martin.weidner@economics.ox.ac.uk}. \newline We are grateful to Isaiah Andrews, St\'ephane Bonhomme, Bo Honor\'e, Ryo Okui, and Peter Schmidt for comments, and to Arthur Lewbel and three referees for feedback on an earlier version. We also greatly appreciate the help of Ryo Okui and Mototsugu Shintani in providing us with the data used in our empirical illustration. Jochmans gratefully acknowledges financial support from the European Research Council through grant ERC-2016-StG-715787-MiMo. Weidner gratefully acknowledges support from the Economic and Social Research Council through the ESRC Centre for Microdata Methods and Practice grant RES-589-28-0001 and from the European Research Council grants ERC-2014-CoG-646917-ROMIA and ERC-2018-CoG-819086-PANEDA.} \\  Nuffield College and Department of Economics \\ University of Oxford
    }
      \date{December 2021}
  \maketitle
} \fi

\if0\blind
{
  \bigskip
  \bigskip
  \bigskip
  \begin{center}
    {\bf {\large INFERENCE ON A DISTRIBUTION FROM NOISY DRAWS}}
\end{center}
  \medskip
} \fi

\vspace{-1cm}
\begin{abstract}
\noindent
We consider a situation where the distribution of a random variable is being estimated by the empirical distribution of noisy measurements of that variable. This is common practice in, for example, teacher value-added models and other fixed-effect models for panel data. We use an asymptotic embedding where the noise shrinks with the sample size to calculate the leading bias in the empirical distribution arising from the presence of noise. The leading bias in the empirical quantile function is equally obtained. These calculations are new in the literature, where only results on smooth functionals such as the mean and variance have been derived. We provide both analytical and jackknife corrections that recenter the limit distribution and yield confidence intervals with correct coverage in large samples. Our approach can be connected to corrections for selection bias and shrinkage estimation and is to be contrasted with deconvolution. Simulation results confirm the much-improved sampling behavior of the corrected estimators. An empirical illustration on heterogeneity in deviations from the law of one price is equally provided. 
\end{abstract}



\newpage

\spacingset{1.45}

\renewcommand{\theequation}{\arabic{section}.\arabic{equation}}  \setcounter{equation}{0}

\section{Introduction}
Let $\theta_1,\ldots,\theta_n$ be a random sample from a distribution $F$ that is of interest. Suppose that we only observe noisy measurements of these variables, say $\vartheta_1,\ldots,\vartheta_n$. A popular approach is to do inference on $F$ and its functionals using the empirical distribution of $\vartheta_1,\ldots,\vartheta_n$. This is common practice when analyzing panel data with heterogenous coefficients. In the literature on student achievement, for example, $\theta_i$ is a teacher effect, $\vartheta_i$ is an estimator of it obtained from data on student test scores, and we care about the distribution of teacher value-added (see, e.g., \citealt{JacksonRockoffStaiger2014} for an overview). In the same vein, \cite{Guvenen2009}, \cite{BrowningEjrnaesAlvarez2010}, and \cite{MagnacRoux2019} estimate heterogenous earning profiles, while  \cite{AhnChoiGaleKariv2014} find substantial heterogeneity in ambiguity aversion
In a nonlinear fixed-effect model, the marginal effect is heterogenous across units and interest lies in the distribution of these effects as well as its functionals (\citealt{Chamberlain1984}, \citealt{HahnNewey2004}).
Although the plug-in approach is popular, using $\vartheta_1,\ldots,\vartheta_n$ rather than $\theta_1,\ldots,\theta_n$ introduces bias that is almost entirely ignored in practice.
\cite{BarrasGagliardiniScaillet2018}, who are interested in the distribution of the skill of fund managers, find that not accounting for bias leads to substantial overestimation of tail mass and misses to pick up the substantial asymmetry in the skill distribution.

We analyze the properties of the plug-in estimator of $F$ in a location-scale setting where
$$
\vartheta_i = \theta_i + \frac{\sigma_i}{\sqrt{m}} \, \varepsilon_i,
\qquad
\varepsilon_i \, \vert \,  (\theta_i,\sigma_i^2) \sim \mathrm{i.i.d.}~(0,1),
$$
where $m$ is a parameter that grows with $n$. As the variance of the (heteroskedastic) noise is $\sigma_i^2/m$, this device shrinks the noise as the sample size grows. This is a very natural asymptotic embedding in settings where $\vartheta_i$ is an estimator of $\theta_i$ obtained from a sample of size $m$, as in a panel data setting or meta-analysis \citep{Vivalt2015}. It is related to, yet different from, an approach based on small measurement-error approximations as in \cite{Chesher1991,Chesher2017},\footnote{
\cite{Chesher1991} provides expansions for densities, while we focus on distribution and quantile functions.
\cite{Chesher2017} discusses the impact of noise in the explanatory variables in a quantile-regression model; this is a different setup than the one considered here. \cite{EvdokimovZeleneev2020} use our device of measurement error that shrinks with the sample size to correct inference in generalized method-of-moment problems.
} and has precedent in the analysis of fixed-effect models for panel data, although for different purposes, as discussed in more detail below (see, e.g., \citealt{HahnKuersteiner2002} and \citealt{AlvarezArellano2003}).

\cite{Efron2011} essentially entertains the homoskedastic setting with normal noise, where
$$
\vartheta_i \vert \, \theta_i \sim N(\theta_i,\sigma^2/m),
$$
and defines selection bias as the tendency of the $\vartheta_i$'s associated with the (in magnitude) largest $\theta_i$'s to be larger than their corresponding $\theta_i$. He proposes to deal with selection bias by using the well-known Empirical Bayes estimator of \cite{Robbins1956}, which here is equal to
$$
\vartheta_i + \frac{\sigma^2}{m} \, \nabla^1 \log p(\vartheta_i),
$$
where $p$ is the marginal density of the $\vartheta_i$ and $\nabla^1$ denotes the first-derivative operator.  For example, when $\theta_i\sim N(0,\psi^2)$ this expression then yields the (infeasible) shrinkage estimator
$$
\left(1-\frac{\sigma^2/m}{\sigma^2/m+\psi^2} \right) \, \vartheta_i,
$$
a parametric plug-in estimator of which would be the \cite{JamesStein1961} estimator. More generally, non-parametric implementation would also require estimation of $p$ and its first derivative. Shrinkage to the overall mean (in this case zero) is intuitive, as selection bias essentially manifests itself through the tails of the empirical distribution of the $\vartheta_i$ being too thick.\footnote{The same shrinkage factor is applied to each $\vartheta_i$, a consequence of the noise being homoskedastic. How to deal with heteroskedastic noise in an Empirical Bayes framework is not obvious. Discussion and a recent contribution can be found in and \cite{WeinsteinMaBrownZhang2018}.}
Shrinkage is commonly-applied in empirical work (see, e.g., \citealt{Rockoff2004};  \citealt{ChettyFriedmanRockoff2014}).
It should be stressed, though, that, while shrinkage improves on $\vartheta_1,\ldots,\vartheta_n$ in terms of estimation risk, it does not lead to preferable estimators of the distribution $F$ or its moments. 

The approach taken here is different from \cite{Efron2011}. Without making parametric assumptions on $F$, we calculate the (leading) bias of the naive plug-in estimator of the distribution,
$$
\hat{F}(\theta) := n^{-1} \sum_{i=1}^n \mathrm{1}\lbrace \vartheta_i \leq \theta \rbrace.
$$
This calculation allows to construct estimators that correct for the bias directly.
In the James-Stein problem, where $\theta_i \sim N(\eta,\psi^2)$, for example, the bias under homoskedastic noise equals
$$
- \frac{\theta-\eta}{2} \, \frac{\sigma^2/\psi^2}{m} \phi\left(\frac{\theta-\eta}{\psi}\right) + o(m^{-1}).
$$
Thus, the empirical distribution is indeed upward biased in the left tail and downward biased in the right tail. A bias order of $m^{-1}$ implies incorrect coverage of confidence intervals unless $n/m^2\rightarrow 0$. We present plug-in and jackknife estimators of the leading bias and show that the bias-corrected estimators are asymptotically normal with zero mean and variance $F(\theta)\, (1-F(\theta))$ as long as $n/m^3\rightarrow0$. 
So, bias correction is preferable to the naive plug-in approach for typical data sizes encountered in practice, where $m$ tends to be quite small relative to $n$.
We also provide corresponding bias-corrected estimators of the quantile function of $F$. %

If the distribution of $\sigma_i\, \varepsilon_i$ is fully known, recovering $F$ is a (generalized) deconvolution problem that can be solved for fixed $m$. Deconvolution-based estimators are well studied (see, e.g., \citealt{CarrollHall1988} and \citealt{DelaigleMeister2008}). However, they have a very slow rate of convergence and it is well documented that they can behave quite poorly in small samples.\footnote{There are also solutions to the measurement-error problem based on repeated measurements (or instrumental variables), coupled with suitable independent restrictions (see, for example, \citealt{HorowitzMarkatou1996}, \citealt{LiVuong1998}, \citealt{Hu2008}, \citealt{HuSchennach2008}, and \citealt{BonhommeJochmansRobin2016b,BonhommeJochmansRobin2016a}). These can be useful alternatives in static models for panel data, where the object of interest is the distribution of the random intercept, as in the work of  \cite{HorowitzMarkatou1996}, for example.} In response to this, \cite{Efron2016} has recently argued for a return to a more parametric approach. Our approach delivers intuitive estimators that enjoy the usual parametric convergence rate and are numerically well behaved. 
Although it does not deliver a fixed-$m$ consistent estimator, bias correction further ensures that size-correct inference can be performed, provided that $n/m^3$ is small. It is not clear how to conduct inference based on deconvolution estimators.

Working out the statistical properties of $\hat{F}$ (and of its quantile function) is non-trivial because $\hat{F}$ is a non-smooth function of the data $\vartheta_1,\ldots,\vartheta_n$. As such, the approach taken here is different from, and complementary to, recent work on estimating average marginal effects in panel data models, which only looks at smooth functionals such as the mean and variance (see, e.g., \citealt{FernandezValLee2013}; \citealt{OkuiYanagi2016}). The impact of noise on smooth transformations of the $\vartheta_i$ can be handled using conventional methods based on Taylor-series expansions. We contrast such an approach with our derivations below. How to perform inference on the  quantiles of marginal effects in nonlinear panel models is a long-standing open question \citep{DhaeneJochmans2015}, and the current work can be seen as a first step in that direction.

In work contemporaneous to our own, \cite{OkuiYanagi2018} derive the bias of a kernel-smoothed estimator of $F$ and its derivative.
Such smoothing greatly facilitates the calculation of the bias, making it amenable to conventional analysis. However, it also introduces additional bias terms that require much stronger moment conditions as well as further restrictions on the relative growth rates of $n$, $m$, and the bandwidth that governs the smoothing. Nevertheless, the (leading) bias term obtained in \citet[Theorem 3]{OkuiYanagi2018} coincides with ours in Proposition~\ref{prop:cdf} below. Additional discussion on and comparison between the two different approaches is given in \citet[p.~169--170]{OkuiYanagi2018}.

\section{Large-sample properties of plug-in estimators}
\label{sec:EandE}

Let $F$ be a univariate distribution on the real line. We are interested in estimation of and inference on $F$ and its quantile function
$
q(\tau) : = \inf_{\theta} \lbrace \theta: F(\theta) \geq \tau \rbrace.
$
If a random sample $\theta_1,\ldots,\theta_n$ from $F$ would be available this would be a standard problem. 
We instead consider the situation where $\theta_1,\ldots,\theta_n$ themselves are unobserved and we observe noisy measurements $\vartheta_1,\ldots,\vartheta_n$, with variances $\sigma_1^2/m,\ldots,\sigma_n^2/m$ for a positive real number $m$ which, in our asymptotic analysis below, will be required to grow with $n$. We assume the following.

\begin{assumption} \label{ass:basicMOM}
The variables $(\theta_i, \sigma_i^2, \vartheta_i)$ are i.i.d.~across $i$, with
\begin{align*}
   E(\vartheta_i \, \vert \, \theta_i,\sigma_i^2 ) &= \theta_i \, ,
  &
   E((\vartheta_i-\theta_i)^2 \, \vert \, \theta_i,\sigma_i^2 ) &= \frac{\sigma_i^2} m \, ,
\end{align*}
and $\sigma^2_i \in [\underline{\sigma}^2,\overline{\sigma}^2]\subset (0,\infty)$ for all $i$.
\end{assumption}

\noindent
Our setup reflects a situation where the noisy measurements  $\vartheta_1,\ldots,\vartheta_n$ converge in squared mean to $\theta_1,\ldots,\theta_n$ at the rate $m^{-1}$. A leading case is the situation where $\vartheta_i$ is an estimator of $\theta_i$ obtained from a sample of size $m$ that converges at the parametric rate.\footnote{
Everything to follow can be readily modified to different convergence rates as well as to  the case where 
$$
\mathrm{var}(\vartheta_i \vert \, \theta_i,\sigma_i^2)= \sigma_i^2/m_i ,
$$
with $m_i := p_i m$ for a random variable $p_i\in (0,1]$. It suffices to redefine $\sigma_i^2$ as $\sigma_i^2 / p_i$. When the $\vartheta_i$ represent estimators this device allows for the sample size to vary with $i$. For example, in a panel data setting, it would cover unbalanced panels under a missing-at-random assumption. Further, the requirement that $\vartheta_i$ is unbiased can be relaxed to allow for standard non-linearity bias of order $m^{-1}$. We do not do this here as it is possible quite generally to reduce the bias down to $O(m^{-3/2})$, for example via a jackknife or bootstrap correction, making it negligible in our analysis below. Furthermore, the split-sample jackknife approach to bias correction that we discuss below would automatically take care of this additional $m^{-1}$ bias without modification.
}
We allow $\theta_i$ and $\sigma_i^2$ to be correlated, implying that the noise $\vartheta_i - \theta_i$ is not independent
of $\theta_i$. Hence, we allow for measurement error to be non-classical. Recovering the distribution of $\theta_i$ from a sample of $(\vartheta_i,\sigma_i^2)$ is, therefore,
not a standard deconvolution problem.

It is common to estimate $F(\theta)$ by 
$$
\hat{F}(\theta) : = 
n^{-1} \sum_{i=1}^n \mathrm{1}\lbrace \vartheta_i \leq \theta \rbrace,
$$
the empirical distribution of the $\vartheta_i$ at $\theta$. 
As we will show below, under suitable regularity conditions, such plug-in estimators are consistent and asymptotically normal as $n\rightarrow\infty$ provided that $m$ grows with $n$ so that $n/m^2$ converges to a finite constant. The use of $\vartheta_1,\ldots,\vartheta_n$ rather than $\theta_1,\ldots,\theta_n$ introduces bias of the order $m^{-1}$, in general. This bias implies that test statistics are size distorted and the coverage of confidence sets is incorrect unless $n/m^2$ converges to zero.

The bias problem is easy to see (and fix) when interest lies in smooth functionals of $F$,
$$
\mu:=E(\varphi(\theta_i)),
$$
for a (multiple-times) differentiable function $\varphi$. 
An (infeasible) plug-in estimator based on $\theta_1,\ldots,\theta_n$ would be 
$$
\tilde{\mu} : = n^{-1} \sum_{i=1}^n \varphi(\theta_i).
$$
Clearly, this estimator is unbiased and satisfies 
$
\tilde{\mu}\overset{a}{\sim} N(\mu,\sigma_\mu^2/n)
$
as soon as $\sigma^2_\mu :=\mathrm{var}(\varphi(\theta_i))$ exists.
For the feasible plug-in estimator of $\mu$,
$$
\hat{\mu} := n^{-1} \sum_{i=1}^n \varphi(\vartheta_i),
$$
under standard regularity conditions, a Taylor-series expansion of $\varphi_i(\vartheta_i)$ around $\theta_i$ yields 
$$
E(\hat{\mu}-\mu) = \frac{b_\mu}{m} + O(m^{-3/2}), \qquad
b_\mu : = 
\frac{E(\nabla^2\varphi(\theta_i) \, \sigma_i^2)}{2},
$$
and
$$
\mathrm{var}(\hat{\mu}) = \frac{\sigma^2_\mu}{n} + O\left(n^{-1}m^{-1} \right).
$$
Hence, letting $z\sim N(0,1)$, we have
$$
\frac{\hat{\mu}-\mu}{\sigma_\mu/\sqrt{n}} \overset{a}{\sim} z + \sqrt{\frac{n}{m^2}} \frac{b_\mu}{\sigma_\mu} \sim N(c\, b_\mu/\sigma_\mu,\sigma_\mu^2) ,
$$
as $n/m^2\rightarrow c^2<\infty$ when $n,m\rightarrow\infty$. The noise in $\vartheta_1,\ldots,\vartheta_n$ introduces bias unless $\varphi$ is linear.
It can be corrected for by subtracting a plug-in estimator of $b_\mu/m$ from $\hat{\mu}$. Doing so, again under regularity conditions, delivers and estimator that is asymptotically unbiased as long as $n/m^3\rightarrow 0$.

\subsection{Distribution function}
The machinery from above cannot be applied to deduce the bias of $\hat{F}$ as it is a step function and, hence, non-differentiable. We will derive its leading bias under the following conditions. To state them, we let
$$
\varepsilon_i := \frac{\vartheta_i-\theta_i}{\sigma_i/\sqrt{m}}
$$
and write $f$ for the density function of $F$.

\begin{assumption} \label{ass:cdf}
The variables $\varepsilon_i$ are independent of $(\theta_i,\sigma_i^2)$, their distribution is absolutely continuous and has finite fourth-order moment. The function $f$ is three times differentiable with uniformly bounded derivatives, and one of the following two sets of conditions holds:

\medskip\noindent
A.~The function $E(\sigma_i^{p+1}\vert \theta_i=\theta)$ is $p$-times differentiable for $p=1,2$, the joint density of $(\theta_i, \sigma_i)$ exists, the conditional density function of $\theta_i$ given $\sigma_i$ is twice differentiable with respect to $\theta_i$ and the derivatives are bounded in absolute value by a function $e(\sigma_i)$ such that $E(e(\sigma_i) ) < \infty$.

\medskip\noindent
B.~There exists a deterministic function $\sigma$ so that $\sigma_i=\sigma(\theta_i)$ for all $i$; and (ii) $\sigma$ is three times differentiable and has uniformly-bounded derivatives.

\end{assumption}	

\noindent
Assumption \ref{ass:cdf} imposes smoothess on certain densities and conditional expectations but not on the estimator of $F$.

Define the function
$$
\beta(\theta) : = \frac{E(\sigma_i^2 \vert \theta_i = \theta) \, f(\theta)}{2},
$$
which is well-behaved under Assumption \ref{ass:cdf}, and let
$$
b_F(\theta) : =  \beta^\prime(\theta)
$$
be its derivative. We also introduce the covariance function
$$
\sigma_F(\theta,\theta^\prime) : = 
F(\theta \wedge \theta^\prime ) - F(\theta) \, F(\theta^\prime),
$$
where we use $\theta \wedge \theta^\prime $ to denote $\min\lbrace\theta,\theta^\prime\rbrace$.
Proposition \ref{prop:cdf} summarizes the large-sample properties of $\hat{F}$.

\begin{proposition} \label{prop:cdf}
Let Assumptions \ref{ass:basicMOM} and \ref{ass:cdf} hold. Then, as $n,m \rightarrow \infty$,
$$
E(\hat{F}(\theta)) - F(\theta) = \frac{b_F(\theta)}{m} + O(m^{-3/2}),
\qquad
\mathrm{cov}\big(\hat{F}(\theta),\hat{F}(\theta^\prime)\big)
=
\frac{\sigma_F(\theta,\theta^\prime)}{n} + O(n^{-1}m^{-1}),
$$	
where the order of the remainder terms is uniform in $\theta$. 
If furthermore $n/m^2\rightarrow c\in [0,+\infty)$, then %
$$
\sqrt{n} \left(\hat{F}(\theta)-F(\theta) - \frac{b_F(\theta)}{m}\right)
\rightsquigarrow
\mathbb{G}_F(\theta) ,
$$
where $\mathbb{G}_F(\theta)$ is a mean zero Gaussian process with covariance function $\sigma_F(\theta_1,\theta_2)$.
\end{proposition}

\begin{proof}
The proof is in Appendix A.
\end{proof}	

\noindent
To illustrate the result suppose that $\sigma_i^2$ is independent of $\theta_i$ and that $\theta_i$ has density function
$$
f(\theta) = \frac{1}{\psi} \phi\left( \frac{\theta-\eta}{\psi} \right),
$$
as in the \cite{JamesStein1961} problem. Letting $\sigma^2$ denote the mean of the $\sigma_i^2$ an application of Proposition~\ref{prop:cdf} yields
$$
b_F(\theta) = - \frac{\theta-\eta}{2} \, \frac{\sigma^2}{\psi^2} \, \phi\left( \frac{\theta-\eta}{\psi} \right).
$$
Thus, $\hat{F}(\theta)$ is upward biased when $\theta<\eta$ and is downward biased when $\theta>\eta$. This finding is a manifestation of the phenomenon of regression to the mean (or selection bias, or the winner's curse; see \citealt{Efron2011}). It implies that the empirical distribution tends to be too disperse.

\subsection{Quantile function}
The bias in $\hat{F}$ translates to bias in estimators of the quantile function. 
A natural estimator of the quantile function is the left-inverse of $\hat{F}$. With this definition, the plug-in estimator of the $\tau$th-quantile is
$$
\hat{q}(\tau)  
:= \vartheta_{(\lceil \tau n \rceil)} ,
$$
where $\vartheta_{(\lceil \tau n \rceil)}$ is the $\lceil \tau n \rceil$th order statistic
of our sample, where $\lceil a \rceil $ delivers the smallest integer at least as large as $a$. 

To calculate the leading bias in $\hat{q}(\tau)$ observe that it is an (approximate) solution to the empirical moment condition
$$
\hat{F}(q) - \tau = 0
$$
(with respect to $q$).
From Proposition~\ref{prop:cdf} we know that
$$
E(\hat{F}(q(\tau))) - \tau =  \frac{b_F(q(\tau))}{m} + O(m^{-3/2}) ,
$$
uniformly in $\tau$, so the moment condition that defines the estimator $\hat{q}(\tau)$ is biased. Letting
$$
b_q(\tau) : = 
-\frac{b_F(q(\tau))}{f(q(\tau))},
\qquad
\sigma_q^2(\tau) : =  
\frac{\tau(1-\tau)}{f(q(\tau))^2},
$$
we obtain the following result.

\begin{proposition}\label{prop:quantile}
Let the Assumptions~\ref{ass:basicMOM} and \ref{ass:cdf} hold.
For $\tau \in (0,1)$, assume that $f>0$ in a neighborhood of $q(\tau)$.
Then, 
$$
   \sqrt{n}\left( {\hat{q}(\tau)} - q(\tau)  -  \frac{b_q(\tau)}{m} \right)
   \overset{d}{\rightarrow} N(0, \sigma_q^2(\tau)),
$$	
as $n,m \rightarrow \infty$ with $n/m^2 \rightarrow c\in [0,+\infty)$.
\end{proposition}

\begin{proof}
The proof is in Appendix A.
\end{proof}

\noindent
As an example, when $\theta_i\sim N(\eta,\psi^2)$, independent of $\sigma_i^2$, we have
$$
b_q(\tau) = \frac{\sigma^2/\psi^2}{2} \, (q(\tau)-\eta),
$$
which, in line with our discussion on regression to the mean above, is positive for all quantiles below the median and negative for all quantiles above the median. The median itself is, in this particular case, estimated without plug-in bias of order $m^{-1}$. It will, of course, still be subject to the usual $n^{-1}$ bias arising from the nonlinear nature of the estimating equation.

\section{Estimation and inference}

Propositions \ref{prop:cdf} and \ref{prop:quantile} complement the existing results on the bias in smooth functionals (\citealt{FernandezValLee2013}; \citealt{OkuiYanagi2016}) of the distribution of heterogenous parameters in panel data models. Our calculations confirm that the order of the bias in the empirical distribution and in the quantile function is of the same order as in the smooth case, $m^{-1}$.

\subsection{Split-panel jackknife estimation}
Importantly, our results validate a traditional jackknife approach to bias correction as in \cite{HahnNewey2004} and \cite{DhaeneJochmans2015}. Such an approach exploits the fact that the bias is proportional to $m^{-1}$ and is based on re-estimating $\theta_1,\ldots,\theta_n$ from subsamples. The simplicity of such a method makes it very useful in panel data applications, for example. 

To illustrate how the jackknife would work here, consider a stationary (balanced) $n\times m$ panel. Let $\vartheta_{i,m_1}$ be an estimator of $\theta_i$ constructed from the $n\times m_1$ subpanel consisting of the first $m_1$ cross sections only.
Then
$$
\hat{F}_{m_1}(\theta):= n^{-1} \sum_{i=1}^n \mathrm{1}\lbrace \vartheta_{i,m_1}\leq \theta \rbrace
$$
is the plug-in estimator of $F(\theta)$ based on this subpanel alone. From Proposition \ref{prop:cdf} it follows that
$$
E(\hat{F}_{m_1}(\theta))
=
F(\theta)
+
\frac{b_F(\theta)}{m_1} + O(m_1^{-3/2}).
$$
Using the remaining  $m_2 := m - m_1$ cross sections from the full panel we can equally calculate estimators $\vartheta_{i,m_2}$ and subsequently construct 
$$
\hat{F}_{m_2}(\theta):= n^{-1} \sum_{i=1}^n \mathrm{1}\lbrace \vartheta_{i,m_2}\leq \theta \rbrace,
$$
for which 
$$
E(\hat{F}_{m_2}(\theta))
=
F(\theta)
+
\frac{b_F(\theta)}{m_2} + O(m_2^{-3/2})
$$
follows in the same way. Consequently, 
$$
\tilde{b}_F(\theta)
:=
{m_1} \hat{F}_{m_1}(\theta) 
+
{m_2} \hat{F}_{m_2}(\theta) 
-m \hat{F}(\theta)
$$
is a split-panel jackknife estimator of the leading bias term $b_F(\theta)$. Hence, 
$$
\tilde{F}(\theta):= \hat{F}(\theta) - \frac{\tilde{b}_F(\theta)}{m}.
$$
is a nonparametric bias-corrected estimator. 

A jackknife estimator of the quantile function can be defined in the same way. Moreover, let $\vartheta_{(\lceil \tau n \rceil), m_1}$ and $\vartheta_{(\lceil \tau n \rceil), m_2}$ be the $\lceil \tau n \rceil$ order statistic of the re-estimated quantities in the first and second subsample, respectively. Recall that $\vartheta_{(\lceil \tau n \rceil), m_1}$ is the (approximate) solution to 
$
\hat{F}_{m_1}(q) - \tau = 0, 
$
and so is our estimator of $q(\tau)$ as obtained from the information in the $n\times m_1$ subpanel only. As before, 
$$
\tilde{b}_q(\tau)
:=
{m_1} \vartheta_{(\lceil \tau n \rceil), m_1} 
+
{m_2} \vartheta_{(\lceil \tau n \rceil), m_2}
-m \vartheta_{(\lceil \tau n \rceil)}
$$
is a nonparametric estimator of $b_q(\tau)$ that gives rise to a jackknife bias-corrected estimator of the quantile function. 

The large-sample behavior of these jackknife estimators is  the same as for the analytic corrections in Propositions \ref{prop:WeakConvCorrected} and \ref{prop:Qbias} below.
The split-sample jackknife is simple to implement but requires access to the original data from which $\vartheta_1,\ldots,\vartheta_n$ were computed. This can be infeasible in meta-analysis problems, where each of the $\vartheta_i$ is an estimator constructed from a different data set that need not all be accessible. It can also be complicated in structural econometric models, where $\vartheta_i$ may be the solution to a cumbersome optimization programme that can be time-consuming to solve. We discuss an alternative bias-correction estimator next.

\subsection{Analytic bias correction}
We will formulate regularity conditions for a plug-in estimator of the bias to be consistent under the maintained assumption that the $\sigma_1^2,\ldots, \sigma_m^2$ are known. We conjecture that, under suitable conditions, the results below will continue to go through when the $\sigma_i^2$ are replaced by estimators. 

A bias-corrected estimator based on Proposition~\ref{prop:cdf} takes the form 
$$
\check{F}(\theta) : = \hat{F}(\theta) - \frac{\hat{b}_F(\theta)}{m},
\qquad
\hat{b}_F(\theta) : = -
\frac{(nh^2)^{-1} \sum_{i=1}^n \sigma_i^2 \, k^\prime\left(\frac{\vartheta_i-\theta}{h}\right)}{2},
$$
where $k^\prime$ is the derivative of kernel function $k$ and $h$ is a non-negative bandwidth parameter. Thus, we estimate the bias using standard kernel methods. For simplicity, we will use a Gaussian kernel throughout, so $k^\prime(\eta):=-\eta \, \phi(\eta)$.

We establish the asymptotic behavior of $\check{F}$ under the following  conditions.

\begin{assumption} \label{ass:cdfcorrect}
\mbox{}

\medskip\noindent
(i) The conditional density of $\theta_i$ given $\sigma_i$ is five times differentiable with respect to $\theta_i$ and the derivatives are bounded in absolute value by a function $e(\sigma_i)$ such that
$
E(e(\sigma_i) ) < \infty.
$ 

\medskip\noindent
(ii)  
There exists an integer $\omega>2$, and real numbers $\kappa> 1 + (1-\omega^{-1})^{-1}$ and $\eta>0$ so that $\sup_\theta (1+\lvert\theta\rvert^{\kappa}) \, f(\theta) = O(1)$ and $\sup_\theta (1+\lvert\theta\rvert^{1+\eta}) \, \lvert \nabla^1 b_F(\theta)\rvert = O(1)$, and $\sup_\theta \lvert b_F(\theta)\rvert = O(1)$.

\medskip\noindent
(iii) The density of $\varepsilon$, $g$, satisfies
$
g(\varepsilon) \leq C \, (1+\lvert \varepsilon \rvert)^{-\alpha}
$
for finite constant $C$ and $\alpha\geq \kappa+1$.

\end{assumption}

\noindent
Assumption \ref{ass:cdfcorrect} contains simple smoothness and boundedness requirements on the conditional density of $\theta_i$ given $\sigma_i^2$, as well as tail conditions on the marginal density of the $\theta_i$ and on the bias function $b_F(\theta)$.

We have the following result.

\begin{proposition}
\label{prop:WeakConvCorrected}
Let Assumptions \ref{ass:basicMOM}, \ref{ass:cdf}, and \ref{ass:cdfcorrect} hold and let $\varepsilon:=(3-\omega^{-1})\, \omega^{-1}>0$. If $h = O(m^{-1/2})$, $h^{-1} = O( m^{2/3 - 4/9 \, \varepsilon} )$, and $h^{-1} = O(n)$, as $n\rightarrow\infty$ and $m\rightarrow\infty$ with $n/m^4\rightarrow 0$, then
$$
\sqrt{n} ( \check{F}(\theta)  -   F(\theta) ) \rightsquigarrow  \mathbb{G}_F(\theta)
$$
as a stochastic process indexed by $\theta$, where $\mathbb{G}_F(\theta)$ is a mean zero Gaussian process with covariance function $\sigma_F(\theta_1,\theta_2)$.    
\end{proposition}

\begin{proof}
The proof is in Appendix B.	
\end{proof}

\noindent
The implications of Proposition \ref{prop:WeakConvCorrected} are qualitatively similar to those for smooth functionals discussed above. Indeed, for any fixed $\theta$, it implies that
$$
\check{F}(\theta) \overset{a}{\sim} N(F(\theta), F(\theta)(1-F(\theta))/n)
$$
as $n\rightarrow\infty$ and $m\rightarrow\infty$ with $n/m^4\rightarrow 0$. Thus, the leading bias is removed from $\hat{F}$ without incurring any cost in terms of (asymptotic) precision. Given the correction term, the sample variance of 
$$
\mathrm{1}\lbrace \vartheta_i\leq \theta \rbrace +  \frac{1}{2}\frac{1}{mh^2} \sigma_i^2 \kappa^\prime\left(\frac{\vartheta_i-\theta}{h}\right)
$$
is a more natural basis for inference in small samples than is that of $\mathrm{1}\lbrace \vartheta_i\leq \theta \rbrace$. 
A data-driven way of choosing $h$ is by cross validation. A plug-in estimator of the integrated squared error $\int_{-\infty}^{+\infty} (\check{F}(\theta)-F(\theta))^2 \, d\theta$ (up to multiplicative and additive constants) is 
\begin{equation*}
\begin{split}
v(h) : =	& 
\sum_{i=1}^n \sum_{j=1}^n \frac{\sigma_i^2 \sigma_j^2}{h^2} \, \underline{\phi}^\prime(\vartheta_i,\vartheta_j;h)
 + 
\sum_{i=1}^n \sum_{j\neq i}\frac{\sigma_i^2}{h} 
\left( m \phi^\prime\left(\frac{\vartheta_i-\vartheta_j}{h}\right)
-
\frac{nm}{n-1}  \phi\left(\frac{\vartheta_i-\vartheta_j}{h}\right)
\right) ,
\end{split}
\end{equation*}
where we use the shorthand
$$
\underline{\phi}^\prime(\vartheta_i,\vartheta_j;h)
: =  \frac{1}{4}
\frac{1}{\sqrt{2}h} \phi\left(\frac{\vartheta_i-\vartheta_j}{\sqrt{2}h} \right) \left(\frac{1}{2}-\frac{(\vartheta_i+\vartheta_j)^2}{4h^2} + \frac{\vartheta_i \vartheta_j}{h^2} \right).
$$
See Appendix C for details on the derivation. The cross-validated bandwidth then is $\check{h}:=\arg\min_h v(h)$ on the interval $(0,+\infty)$.

Now turn the bias-corrected estimation of the quantile function. Proposition \ref{prop:quantile} readily suggests a bias-corrected estimator of the form
$$
\hat{q}(\tau) - \frac{\hat{b}_q(\tau)}{m}, 
\qquad
\hat{b}_q(\tau) : = 
- \frac{\hat{b}_F(\hat{q}(\tau))}{\hat{f}(\hat{q}(\tau))},
$$
using obvious notation. While (under suitable regularity conditions) such an estimator successfully reduces bias it has the unattractive property that it requires a non-parametric estimator of the density $f$, which further shows up in the denominator. 

An alternative estimator that avoids this issue is
$$
\check{q}(\tau):= \vartheta_{(\lceil \hat{\tau}^* n \rceil)}, \qquad
\hat{\tau}^* : = 
\tau + \frac{\hat{b}_F(\hat{q}(\tau))}{m} ,
$$
The justification for this estimator comes from the fact that
$E(\hat{F}(q(\tau))) - \tau^* = O(m^{-2})$, where $\tau^* = \tau + {b_F(q(\tau))}/ {m}$,  and its interpretation is intuitive. Given the noise in the $\vartheta_i$ relative to the $\theta_i$, the empirical distribution of the former is too heavy-tailed relative to the latter, and so $\hat{q}(\tau)$ estimates a quantile that is too extreme, on average. Changing the quantile of interest from $\tau$ to $\tau^*$ adjusts the naive estimator and corrects for regression to the mean.

\begin{proposition} \label{prop:Qbias}
Let the assumptions stated in Proposition~\ref{prop:WeakConvCorrected} hold.
For $\tau \in (0,1)$, assume that $f>0$ in a neighborhood of $q(\tau)$.
Then, 
$$
   \sqrt{n}\left( \check{q}(\tau) - q(\tau)   \right)
   \overset{d}{\rightarrow} {N}(0, \sigma_q^2(\tau)),
$$	
as $n,m \rightarrow \infty$ with $n/m^4 \rightarrow 0$.
\end{proposition}	

\begin{proof}
The proof is in Appendix B.	
\end{proof}

\noindent
The corrected estimator has the same asymptotic variance as the uncorrected estimator.  
It is well-known that plug-in estimators of $\sigma_q^2$ can perform quite poorly in small samples (\citealt{MaritzJarrett1978}). Typically, researchers rely on the bootstrap, and we suggest doing so here. Moreover, draw (many) random samples of size $n$ from the original sample $\vartheta_1,\ldots,\vartheta_n$ and re-estimate $q(\tau)$ by the bias-corrected estimator for each such sample. Then construct confidence intervals for $q(\tau)$ using the percentiles of the empirical distribution of these estimates. Note that, again, this bootstrap procedure does not involve re-estimation of the individual $\theta_i$.

\section{Numerical illustrations} \label{sec:illustrations}

\subsection{Simulated data}

To support our theory we provide simulation results for a \cite{JamesStein1961} problem where $\theta_i\sim N(0,\psi^2)$ and we have access to an $n\times m$ panel on independent realizations of the random variable
$$
x_{it} \vert \, \theta_i\sim N(\theta_i,\sigma^2).
$$ 
This setup is a simple random-coefficient model. It is similar to the classic many normal means problem of \cite{NeymanScott1948}.  While their focus was on consistent estimation of the within-group variance, $\sigma^2$, for fixed $m$, our focus is on between-group characteristics and the distribution of the $\theta_i$ as a whole.
We estimate $\theta_i$ by the fixed-effect estimator, i.e., 
$$
\vartheta_i = m^{-1}\sum_{t=1}^m x_{it}.
$$
The sampling variance of $\vartheta_i \vert \theta_i$ is $\sigma^2/m$. Rather than assuming this variance to be known we implement our analytical bias correction using the estimator
$$
s_i^2 := (m-1)^{-1} \sum_{t=1}^m (x_{it}-\vartheta_i)^2 .
$$
We do not make use of the fact that the $\vartheta_i$ are homoskedastic in estimating the noise or in constructing the bias correction. Moreover, the implementation of our procedure is non-parametric in the noise distribution. 

A deconvolution argument implies that 
$$
\vartheta_i \sim N(0, \psi^2 + \sigma^2/m).
$$
Thus, indeed, the empirical distribution of the fixed-effect estimator is too fat-tailed. In particular, the sample variance of $\vartheta_1,\ldots,\vartheta_n$,
$$
\hat{\psi}^2 : = \frac{1}{n-1} \sum_{i=1}^n (\vartheta_i-\overline{\vartheta})^2,
\qquad
\overline{\vartheta} := n^{-1} \sum_{i=1}^n \vartheta_i,
$$
is a biased estimator of $\psi^2$. To illustrate how this invalidates inference in typically-sized data sets we simulated data for $\psi^2=1$ (so $F$ is standard normal) and $\sigma^2=5$. The panel dimensions $(n,m)$ reported on are $(50,3)$, $(100,4)$, and $(200,5)$. Table \ref{tab:JamesStein} shows the bias and standard deviation of $\hat{\psi}^2$ as well as the empirical rejection frequency of the usual two-sided $t$-test for the null that $\psi^2=1$. The nominal size is set to $5\%$. In practice, however, the test rejects in virtually all of the $10,000$ replications. The table provides the same summary statistics for the bias-corrected estimator
$$
\check{\psi}^2 : = {\frac{1}{n-1} \sum_{i=1}^n\left(  (\vartheta_i-\overline{\vartheta})^2-\frac{s_i^2}{m}\right)}.
$$
The adjustment reduces the estimator's bias relative to its standard error and brings down the empirical rejection frequencies to just over their nominal value for the sample sizes considered.

\begin{table}[htbp]
\centering \footnotesize
\caption{Inference on $\psi^2$ in the James-Stein problem from $n\times m$ panel data.}
\begin{tabular}{rrrrrrrrrr}
	\hline\hline
          &       & \multicolumn{2}{c}{bias} & \multicolumn{2}{c}{std} & \multicolumn{2}{c}{se/std} & \multicolumn{2}{c}{size ($5\%$)} \\
		  \hline
    $n$   & $m$     & \multicolumn{1}{c}{$\hat{\psi}^2$}   & \multicolumn{1}{c}{$\check{\psi}^2$}    & \multicolumn{1}{c}{$\hat{\psi}^2$}   & \multicolumn{1}{c}{$\check{\psi}^2$}    & \multicolumn{1}{c}{$\hat{\psi}^2$}   & \multicolumn{1}{c}{$\check{\psi}^2$}    & \multicolumn{1}{c}{$\hat{\psi}^2$}   & \multicolumn{1}{c}{$\check{\psi}^2$ }\\
     50   & 3    & 1.616 & -0.054 & 0.525 & 0.577 & 0.964 & 0.971 & 0.973 & 0.082 \\
    100   & 4    & 1.224 & -0.028 & 0.321 & 0.337 & 0.966 & 0.969 & 0.997 & 0.073 \\
    200   & 5    & 0.989 & -0.010 & 0.199 & 0.205 & 0.985 & 0.985 & 1.000 & 0.062 \\
	\hline\hline
\end{tabular}
\vspace{-.25cm} \floatfoot{ \footnotesize 
Table notes.
$\hat{\psi}^2$ is the plug-in estimator of $\psi^2$. $\check{\psi}^2$ is its (analytically) bias-corrected version constructed using estimators of the variance of the noise distributions. The table reports the bias and standard deviation of these estimators, along with the ratio of the average standard error to the standard deviation and empirical rejection frequencies of a two-sided $t$-test for the null that $\psi^2=1$, which is the value with which the data were generated.
}
\label{tab:JamesStein}
\end{table}

A popular approach in empirical work to deal with noise in $\vartheta_1,\ldots,\vartheta_n$ is shrinkage estimation (see, e.g., \citealt{ChettyFriedmanRockoff2014}). This procedure is not designed to improve estimation and inference of $F$ or its moments, however. In the current setting, the (infeasible, parametric) shrinkage estimator is simply
$$
\left(1 - \frac{\sigma^2/m}{\sigma^2/m+\psi^2} \right) \, \vartheta_i.
$$
Its exact sampling variance is
$$
\left(\frac{\psi^2}{\sigma^2/m+\psi^2} \right) \, \psi^2
=
\psi^2 - \frac{\sigma^2/\psi^2}{m} + o(m^{-1}).
$$
It follows that the sample variance of the shrunken $\vartheta_1,\ldots,\vartheta_n$ has a bias that is of the same order as that in the sample variance of  $\vartheta_1,\ldots,\vartheta_n$. Interestingly, note that, here, this estimator overcorrects for the presence of noise, and so will be underestimating the true variance, $\psi^2$, on average.

The upper two plots in Figures \ref{fig:simulations_n50_m3}, \ref{fig:simulations_n100_m4}, and \ref{fig:simulations_n200_m5} provide simulation results for the distribution function $F$ for the same Monte Carlo designs. The figures deal with the sample sizes $(50,3)$, $(100,4)$, $(200,5)$, respectively. The left plots contain (the average over the Monte Carlo replications of) the analytically bias-corrected estimator (solid blue line), with the bandwidth chosen according to a cross-validation procedure, together with $95\%$ confidence bands placed around in. Each of the plots also provide the average of the naive plug-in estimator (dashed red line), the empirical distribution of the Empirical-Bayes point estimates (dashed-dotted purple line), and the actual standard-normal distribution that is being estimated (solid black line).\footnote{Empirical Bayes was implemented non-parametrically (and correctly assuming homoskedasticity) based on the formula stated in the introduction using a kernel estimator and the optimal bandwidth that assumes knowledge of the normality of the target distribution.} The upper right plots in Figures \ref{fig:simulations_n50_m3}, \ref{fig:simulations_n100_m4}, and \ref{fig:simulations_n200_m5} have the same structure, only now the bias-corrected estimator being plotted is the split-sample jackknife.

The simulations clearly show the substantial bias in the naive estimator. This bias becomes more pronounced relative to its standard error as the sample size grows and, indeed, $\hat{F}$ starts falling outside of the confidence bands (of the bias corrected estimator) as the sample size increases. The Empirical-Bayes estimator is less biased than $\hat{F}$. However, its bias is of the same order and so, as the sample size grows it does not move toward $F$ but, rather, towards $\hat{F}$.\footnote{Recall that the Empirical-Bayes estimator is not designed for inference on $F$ but, in stead, aims to minimize risk in estimating $\theta_1,\ldots,\theta_n$. In terms of RMSE it dominates $\vartheta_1,\ldots,\vartheta_n$. For the three sample sizes considered here, the RMSEs are $1.667$, $1.246$, and $1.000$ for the plug-in estimators and $1.233$, $1.018$, $.874$ for Empirical Bayes.}
The confidence bands of $\check{F}$ and $\tilde{F}$ settle around $F$ as the sample grows. The results also show near identical performance of the split-sample approach and the analytical approach based on our bias formula. Indeed, the curves in the left and right plots are virtually indistinguishable.

\begin{figure}
\centering
\caption{Estimation of $F$ and $q$ in the James-Stein problem from $50\times 3$ panel data.}
\includegraphics[height=.70\textheight]{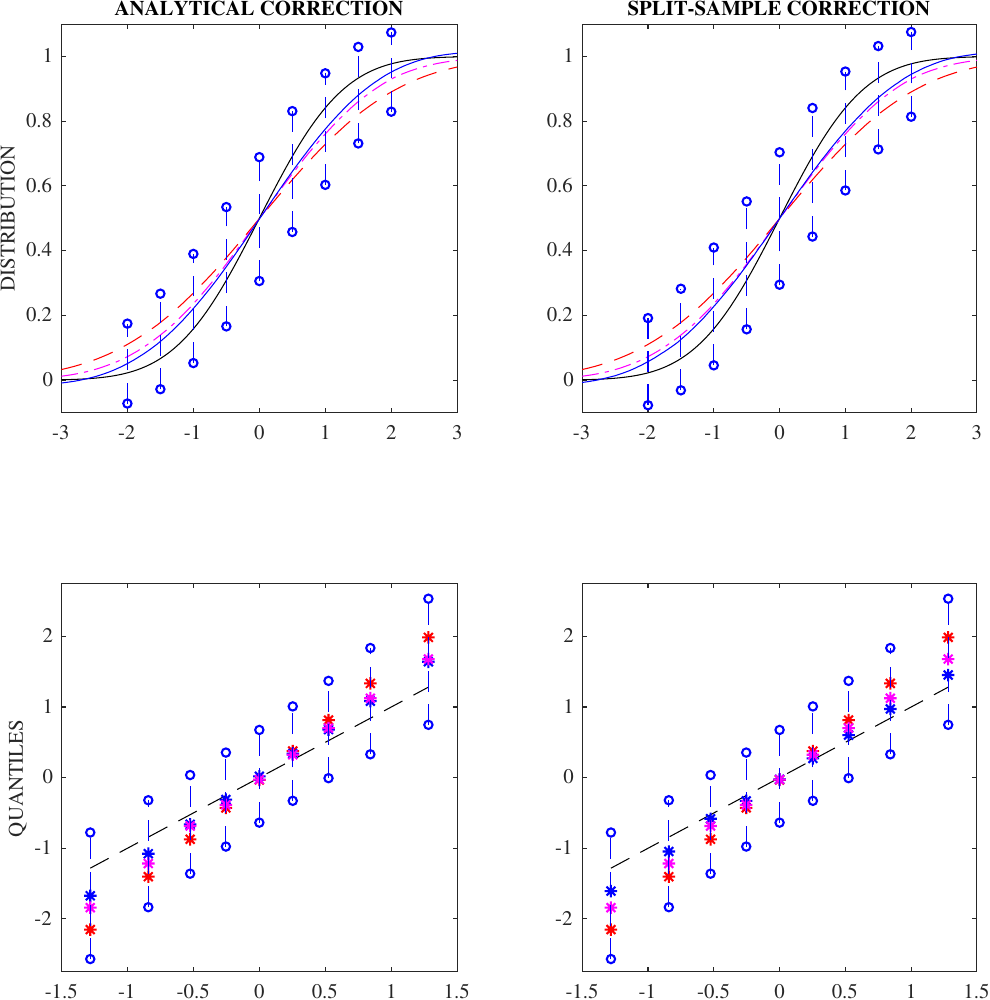}
\vspace{-.25cm} \floatfoot{ \footnotesize 
Figure notes.
The upper plots contain the average (over the Monte Carlo replications) distribution function (full blue line) obtained via analytical bias correction (left plot) and split-sample jackknife estimation (right plot) along with $95\%$ confidence intervals around them at each of the quantiles of $F$ (blue $\circ$).  Each plot also contains the true curve (full black line) and the average of the empirical distribution function of the estimated $\theta_i$ (dashed red line) and of their Empirical Bayes adjustments (dashed-dotted purple line). The lower plots contain corresponding QQ-plots of the average bias-corrected quantile function (blue $*$) at each of the deciles of $F$ together with $95\%$ confidence intervals. The $45^\circ$ line (dashed black line) corresponds to the truth. Average estimates for the naive (red $*$) and Empirical Bayes (purple $*$) estimator are equally pictured.
}
\label{fig:simulations_n50_m3}
\end{figure}

\begin{figure}
\centering
\caption{Estimation of $F$ and $q$ in the James-Stein problem from $100\times 4$ panel data.}
\includegraphics[height=.70\textheight]{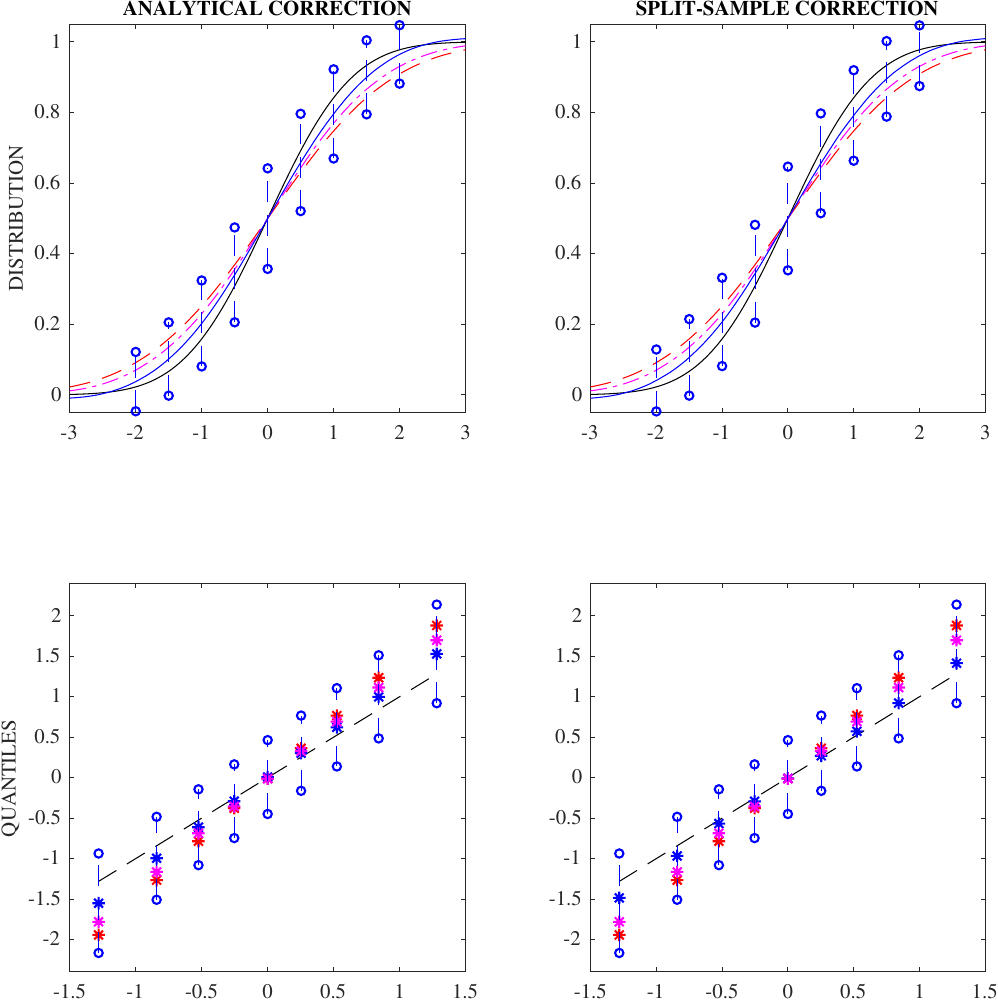}
\vspace{-.25cm} \floatfoot{ \footnotesize 
Figure notes.
The upper plots contain the average (over the Monte Carlo replications) distribution function (full blue line) obtained via analytical bias correction (left plot) and split-sample jackknife estimation (right plot) along with $95\%$ confidence intervals around them at each of the quantiles of $F$ (blue $\circ$).  Each plot also contains the true curve (full black line) and the average of the empirical distribution function of the estimated $\theta_i$ (dashed red line) and of their Empirical Bayes adjustments (dashed-dotted purple line). The lower plots contain corresponding QQ-plots of the average bias-corrected quantile function (blue $*$) at each of the deciles of $F$ together with $95\%$ confidence intervals. The $45^\circ$ line (dashed black line) corresponds to the truth. Average estimates for the naive (red $*$) and Empirical Bayes (purple $*$) estimator are equally pictured.
}
\label{fig:simulations_n100_m4}
\end{figure}

\begin{figure}
\centering
\caption{Estimation of $F$ and $q$ in the James-Stein problem from $200\times 5$ panel data.}
\includegraphics[height=.70\textheight]{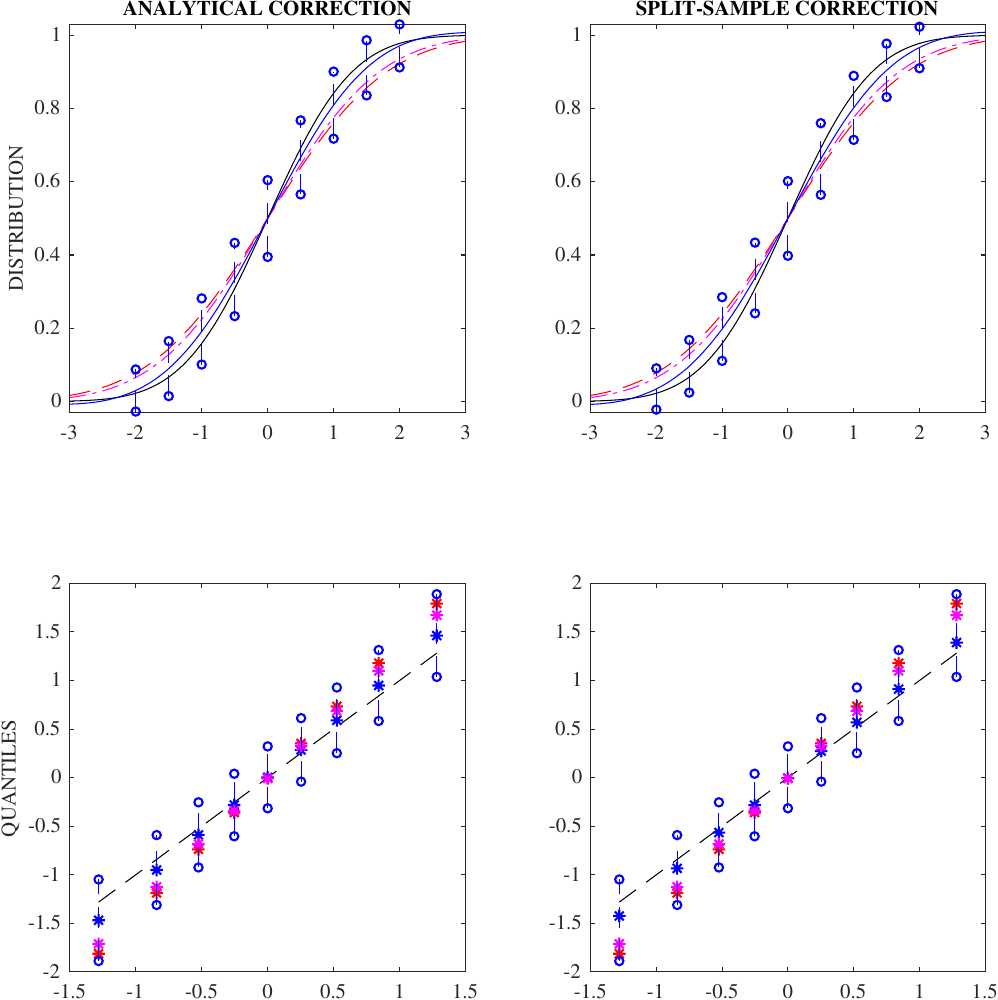}
\vspace{-.25cm} \floatfoot{ \footnotesize 
Figure notes.
The upper plots contain the average (over the Monte Carlo replications) distribution function (full blue line) obtained via analytical bias correction (left plot) and split-sample jackknife estimation (right plot) along with $95\%$ confidence intervals around them at each of the quantiles of $F$ (blue $\circ$).  Each plot also contains the true curve (full black line) and the average of the empirical distribution function of the estimated $\theta_i$ (dashed red line) and of their Empirical Bayes adjustments (dashed-dotted purple line). The lower plots contain corresponding QQ-plots of the average bias-corrected quantile function (blue $*$) at each of the deciles of $F$ together with $95\%$ confidence intervals. The $45^\circ$ line (dashed black line) corresponds to the truth. Average estimates for the naive (red $*$) and Empirical Bayes (purple $*$) estimator are equally pictured.
}
\label{fig:simulations_n200_m5}
\end{figure}

The reduction in bias in our estimators of $F$ is again sufficient to bring the empirical size of tests in line with their nominal size. To see this Table \ref{tab:normalsimulation} provides  empirical rejection frequencies of two-sided tests at the $5\%$ level for $F$ at each of its deciles using both $\hat{F}$ and $\check{F}$. The rejection frequencies based on the naive estimator are much too high for all sample sizes and deciles and get worse as the sample gets larger. Empirical size is much closer to nominal size after adjusting for noise, and this improvement is observed at all deciles of the distribution.

\begin{table}[htbp]
\centering \footnotesize
\caption{Inference on $F$ in the James-Stein problem from $n\times m$ panel data.}
\begin{tabular}{rrrrrrrrrr}
\hline\hline
$\tau$ & 	\multicolumn{1}{c}{$.1$} & \multicolumn{1}{c}{$.2$} & \multicolumn{1}{c}{$.3$} & \multicolumn{1}{c}{$.4$} & \multicolumn{1}{c}{$.5$} & \multicolumn{1}{c}{$.6$} & \multicolumn{1}{c}{$.7$} & \multicolumn{1}{c}{$.8$} & \multicolumn{1}{c}{$.9$} \\
& \multicolumn{9}{c}{$(n,m)=(50,3)$} \\
$\hat{F}$   & 
0.4814 & 0.5518 & 0.3695 & 0.1530 & 0.0681 & 0.1598 & 0.3801 & 0.5610 & 0.4828 \\
$\check{F}$ &
0.0600 & 0.0928 & 0.1039 & 0.0785 & 0.0563 & 0.0745 & 0.1029 & 0.0891 & 0.0628 \\
& \multicolumn{9}{c}{$(n,m)=(100,4)$} \\	
$\hat{F}$   & 						
0.6962 & 0.7304 & 0.5564 & 0.2280 & 0.0566 & 0.2312 & 0.5586 & 0.7352 & 0.7034 \\
$\check{F}$   & 
0.0608 & 0.0848 & 0.0920 & 0.0664 & 0.0494 & 0.0734 & 0.0932 & 0.0782 & 0.0532 \\
& \multicolumn{9}{c}{$(n,m)=(200,5)$} \\
$\hat{F}$   & 						
0.926 & 0.902 &	0.7634 &0.3288 &0.0576 & 0.3212 & 0.7646 & 0.903 & 0.9146 \\
$\check{F}$   & 
0.0536 & 0.0828 & 0.0996 & 0.0770 & 0.0496 & 0.0792 & 0.0978 & 0.0780 & 0.0554 \\
	\hline\hline
\end{tabular}
\vspace{-.25cm} \floatfoot{ \footnotesize 
Table notes.
$\hat{F}$ is the empirical distribution of the $\vartheta_i$. $\check{F}$ is its (analytically) bias-corrected version constructed using estimators of the variance of the noise distributions. The table provides, for several combinations of $(n,m)$, rejection frequencies of the associated two-sided tests of the null that $F(\Phi^{-1}(\tau))=\tau$ for a range of different quantiles $\tau$; the data were generated with $F$ set to the standard-normal distribution function.
}
\label{tab:normalsimulation}
\end{table}

The lower two plots in Figures \ref{fig:simulations_n50_m3}, \ref{fig:simulations_n100_m4}, and \ref{fig:simulations_n200_m5} provide corresponding simulation results for estimators of the deciles of $F$. The presentation is constructed around a QQ-plot of the standard normal, pictured as the black dashed-dotted line in each plot. Along the QQ-plot, the average (over the Monte Carlo replications) of the naive estimator (red), Empirical Bayes (purple), and the  bias-corrected quantiles (blue) are shown by $*$ symbols. Again, the left plots deal with the analytical correction and the right plots show results for the split-sample approach. Confidence intervals around the corrected estimators (in blue,-o) are also again provided. Like the naive estimator, the Empirical Bayes estimators reported are the appropriate order statistics of $\vartheta_1,\ldots,\vartheta_n$, after shrinkage has been applied to each. 
Visual inspection reveals that the results are in line with those obtained for the distribution function. As the sample size grows, only $\check{q}$ successfully adjusts for bias arising from estimation noise in $\vartheta_1,\ldots,\vartheta_n$. Here, the split-sample correction is slightly more effective than our analytical approach. 

\subsection{Empirical illustration}
We use quarterly panel data on a set of 48 consumer price index items in 52 US cities. The data span the period 1990--2007, yielding 72 time series observations. They were used by \cite{ParsleyWei1996}, \cite{CruciniShintaniTsuruga2015}, and \cite{OkuiYanagi2016,OkuiYanagi2018} to investigate the cross-sectional heterogeneity in deviations from the law of one price. Let $p_{cit}$ be the price of item $i$ in city $c$ at time $t$ and define the random variable
$$
x_{cit} = \log \left( \frac{p_{cit}}{p_{1it}} \right) = \log (p_{cit}) - \log (p_{1it})
$$
for all $(52-1)\times 48=2448$ city/item combinations apart from the reference city (which here is Albuquerque, New Mexico). For each city/item combination we estimate the mean, standard deviation and first-order autocorrelation of $x_{cit}$ non-parametrically from the time dimension of our panel. Our interest lies in the distribution functions of their population counterparts. We estimate these three distributions by the empirical distributions of the cross-sectional estimates, and then correct for plug-in bias via the split-sample jackknife procedure. Our results complement the analysis of \citet[Figure 1]{OkuiYanagi2018}, which gives corresponding estimates of the associated density functions.

\begin{figure}
\centering
\caption{Deviations from the law of one price}
\includegraphics[width=.97\textwidth]{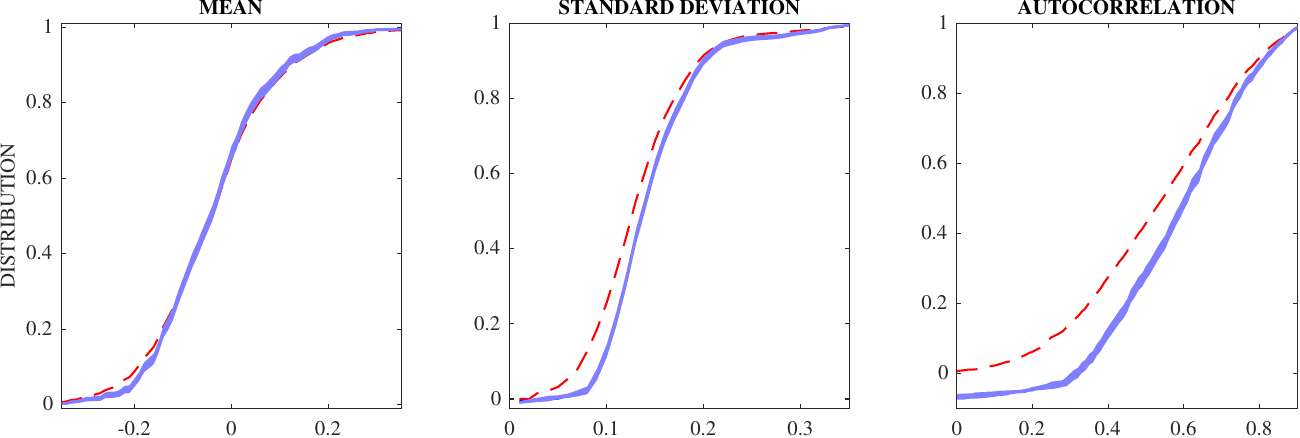}
\vspace{-.25cm} \floatfoot{ \footnotesize 
Figure notes.
\footnotesize The empirical distribution functions of the means (left) standard deviations (middle) and autocorrelations (right) of the time series of $x_{cit} = \log (p_{cit}) - \log (p_{1it})$ for all city/item combinations (dashed red line) along with $95\%$ confidence bands constructed from the split-sample jackknife estimator of each of these distributions (shaded blue region).
}
\label{fig:LOP}
\end{figure}

The results are collected in Figure \ref{fig:LOP}. The plots contain the empirical distribution functions (dashed red line) together with $95\%$ confidence bands based on the split-sample jackknife (shaded blue region). The correction for regression to the mean to the empirical distribution is clearly visible for the mean (left plot). It is also statistically significant, with the tails of the empirical distribution falling out of the confidence region. The sample standard deviation and autocorrelation obtained from the time series are biased estimators and so the empirical distribution function for these parameters (middle and right plot, respectively) suffer from an additional bias that is of the same order of magnitude as is the bias due to estimation noise (see the discussion on Footnote 4). The split-sample jackknife corrects for both these sources of bias automatically. Here, the bias adjustment leads to a pronounced shift of the empirical distribution; the corrected distribution functions all but stochastically dominate the naive plug-in estimators. The differences between the corrected and uncorrected functions are quantitatively large and, given the small standard error, they are also statistically significant.

\section{Conclusions}
In this paper, we have considered inference on the distribution of latent variables from noisy measurements. In an asymptotic embedding where the variance of the noise shrinks with the sample size, we have derived the leading bias in the empirical distribution function of the noisy measurements and suggested both an analytical and a jackknife correction. They provide a simple and numerically stable (approximate) solution to a generalized deconvolution problem that, in addition, yields valid inference procedures. The split-sample jackknife is particularly straightforward to implement and we recommend its use whenever possible.

\appendix

\renewcommand{\theequation}{A.\arabic{equation}} 
\setcounter{equation}{0}
\renewcommand{\thelemma}{A.\arabic{lemma}} 
\setcounter{lemma}{0}
\section{Appendix} 

\noindent
\textbf{Notational convention:} we let $\nabla_p^q\varphi$ denote the $q$th derivative of $\varphi$ with respect to its $p$th argument. We omit the subscript for univariate $\varphi$.

\subsection{Proof of Proposition \ref{prop:cdf}}

The following known result is useful to  prove Proposition~\ref{prop:cdf}.

\begin{lemma}[\citealt{KomlosMajorTusnady1975}] \label{lemma:KMTapprox}
Let $\mathbb{G}_n$ denote the empirical cumulative distribution of an i.i.d.~sample of size n from a uniform distribution on [0,1]. Let $\mathbb{B}_n$ denote a sequence of Brownian bridges. 
Then	
	\begin{align*}
	        \sup_{u \in [0,1]} \left|   
	        \sqrt{n}  \left( \mathbb{G}_n(u)  - u \right)- \mathbb{B}_n(u)   \right| 
	        = O_p(\log(n)/\sqrt{n} ) .
	\end{align*}
\end{lemma}	

\bigskip

\begin{proof}[\bf Proof of Proposition~\ref{prop:cdf}]

We begin with the bias calculation. 
Suppose, first, that Part A of Assumption~\ref{ass:cdf} holds. 
Then $(\theta_i,\sigma_i)$ have a joint density, $h(\theta_i,\sigma_i)$, say. We will denote the marginal density of $\sigma_i$ by $h(\sigma_i)$  and the conditional density of $\theta_i$ given $\sigma_i$ by $h(\theta_i \vert \sigma_i)$.
For any real number $\delta$ let 
$$
G(\theta,\delta) 
: = 
E(\mathrm{1}\lbrace \theta_i + \delta \sigma_i \leq \theta \rbrace)
=
\textstyle{\int}_{\underline{\sigma}}^{\overline{\sigma}} \int_{-\infty}^{\theta-\delta\sigma} h(\vartheta,\sigma) \, d\vartheta \, d\sigma .
$$
Note that $G(\theta,0) = F(\theta)$ and that
\begin{align}
E(\hat{F}(\theta)) 
= 
E(\mathrm{1}\lbrace \vartheta_i \leq \theta \rbrace)
=
E\left(\mathrm{1}\left\lbrace \theta_i + \frac{\varepsilon_ i}{\sqrt{m}} \sigma_i\leq \theta \right\rbrace\right)
=
E\left(G(\theta,\varepsilon_i/\sqrt{m}) \right).
   \label{EhatF}
\end{align}
Assumption~\ref{ass:cdf} implies that $G$ is smooth and differentiable in its second argument.
By definition of the function $e(\sigma_i)$,
\begin{align}
\sup_\theta\sup_\delta \lvert \nabla_2^3 G(\theta,\delta) \rvert =
\sup_\theta\sup_\delta
\left\lvert 
\textstyle{\int}_{\underline{\sigma}}^{\overline{\sigma}} \sigma^3\,  \nabla_1^2 h(\theta-\delta \sigma\vert \sigma) \, h(\sigma) \, d\sigma
\right\rvert
\leq
\textstyle{\int}_{\underline{\sigma}}^{\overline{\sigma}} \sigma^3\,  e(\sigma) \, h(\sigma) \, d\sigma ,
   \label{BoundedDerivativeG}
\end{align}
which equals $E(\sigma_i^3 e(\sigma_i))$ and is finite by assumption. 
Therefore, by \eqref{EhatF} and
a third-order expansion of $G(\theta,\varepsilon_i/\sqrt{m})$ in its second argument around zero we find that
$$
E(\hat{F}(\theta)) 
=
F(\theta)
+
\frac{1}{2}
\frac{\nabla_2^2 G(\theta,0)}{m} \,
+
\frac{1}{6}
\frac{E(\varepsilon_i^3 \, \nabla_2^3 G(\theta,\varepsilon_i^*/\sqrt{m}))}{m^{3/2}}
=
\frac{1}{2}
\frac{\nabla_2^2 G(\theta,0)}{m} + O(m^{-3/2}) ,
$$
where $\varepsilon_i^*$ is some value between zero and $\varepsilon_i$,
and where, in addition to \eqref{BoundedDerivativeG}, 
we have used that $E(\varepsilon_i )= 0$ and $E(\varepsilon_i^2)= 1$ by construction and that $E(\lvert \varepsilon_i \rvert^3)<\infty$ by assumption. 
By direct calculation,
$$
\nabla_2^2 G(\theta,0) =  2 \, b_F(\theta).
$$
Therefore,
$$
E(\hat{F}(\theta)) = F(\theta) + \frac{b_F(\theta)}{m} + O(m^{-3/2}),
$$ 
as claimed.

Suppose, next, that Part B of Assumption~\ref{ass:cdf} holds. 
Then we have a deterministic relationship between $\theta_i$ and $\sigma_i$. We may define $G(\theta,\delta)$ as above but have to take care when Taylor expanding in $\delta$, as the function may be non-continuous. A non-continuity occurs whenever the number of solutions $t$ (on the real line) to the equation $t+\delta \sigma(t) = \theta$ changes. 
However,  at $\delta=0$ the only solution to this equation is $t=\theta$,
and because we assume that the function $\sigma(\theta)$ has uniformly bounded derivative $\sigma^\prime$, 
there  always exists $\eta>0$ such that for all $\delta \in (-\eta,\eta)$
and all real  $\theta$ the equation $t+\delta \sigma(t) = \theta$  has a unique solution in $t$ on the real line.
We denote this solution by $t^*(\theta,\delta)$, that is, we have $t^*(\theta,\delta)+\delta \sigma(t^*(\theta,\delta)) = \theta$.
Using this we find that
for $\delta \in (-\eta,\eta)$ we have
\begin{align*}
     G(\theta,\delta) &= F( t^*(\theta,\delta) ),
     &
     \nabla_2^1{t^*(\theta,\delta)}
     &= -  \frac{ \sigma(t^*(\theta,\delta)) }  {1 + \delta \, \sigma'(t^*(\theta,\delta))} ,
\end{align*}
where the last equation is obtained by taking derivatives of $t^*(\theta,\delta)+\delta \sigma(t^*(\theta,\delta)) = \theta$ with respect to $\delta$ and then solving for the derivative. Because we have that $t^*(\theta,0) = \theta$ we then find
\begin{align*}
       G(\theta,0)  
     &=   F(\theta) ,
   &  
      \nabla_2^1{ G(\theta,0)}
     &=  - \sigma(\theta) f(\theta),
     &  
        \nabla_2^2{ G(\theta,0)}
       &=   2 b_F(\theta).
\end{align*}
Differentiating further we see that $\nabla_2^3 G(\theta,0)$, and $\nabla_2^4 G(\theta,0)$ are functions of the derivatives of $f$ and $\sigma$ up to third order. Our assumption that these derivatives are uniformly bounded implies that
\begin{equation} \label{DerivativesH3} 
\sup_\theta  {\sup_{\delta\in[-\eta, \eta]} } \left\lvert \nabla_2^4 G(\theta,\delta) \right\rvert <\infty.
\end{equation}
The only obstacle that now prevents us from proceeding with an expansion as we did under Assumption~\ref{ass:cdf}.A is that the bound \eqref{DerivativesH3} is restricted to a neighborhood around zero. To complete the derivation of the bias we argue that the restriction that $\delta \in (-\eta,\eta)$ relaxes sufficiently fast as $m$ grows.
We do so as follows. Note, first, that, by Markov's inequality, 
$$
P( \lvert \varepsilon_i \rvert > \eta \sqrt{m}) \leq m^{-2} \frac{E(\varepsilon_i^4)}{\eta^4} = O(m^{-2}).
$$
Then
\begin{equation*}
\begin{split}
E(\hat{F}(\theta)) 
& =	
E\left( G(\theta,\varepsilon/\sqrt{m}) \right)
\\
& = 
E\left(\lbrace  \lvert \varepsilon_i \rvert \leq \eta \sqrt{m} \rbrace \, G(\theta,\varepsilon/\sqrt{m}) \right)
+
E\left(\lbrace  \lvert \varepsilon_i \rvert > \eta \sqrt{m} \rbrace \, G(\theta,\varepsilon/\sqrt{m}) \right)
\\
& = 
E\left(\lbrace  \lvert \varepsilon_i \rvert \leq \eta \sqrt{m} \rbrace \, G(\theta,\varepsilon/\sqrt{m}) \right)
+ O(m^{-2}),
\end{split}
\end{equation*}		
uniformly in $\theta$, because,
$$
\sup_\theta E( \lbrace \lvert \varepsilon_i\rvert >  \eta \, \sqrt{m} \rbrace \,  G(\theta,  {\varepsilon_i}/{\sqrt{m}})  )
\leq 
P(  \lvert\varepsilon_i \rvert >  \eta \, \sqrt{m}   )  =  O(m^{-2}),
$$
noting that $\sup_{\theta} \sup_\delta  G(\theta, \delta)  \leq 1$ by definition of the function $G$. Next, a Taylor expansion of $G$ around $\delta = 0 $ gives
$$
E(\hat{F}(\theta))
=
E(G(\theta,\varepsilon_i/\sqrt{m}))
=
F(\theta)
+
\frac{1}{2} \frac{\nabla_2^2 G(\theta,0) }{m}
+
\frac{1}{6} \frac{\nabla_2^3 G(\theta,0)}{m^{3/2}} + R(\theta) + O(m^{-2}),
$$
where we have used that $F(\theta)=G(\theta,0)$, that $E(\varepsilon_i)=0$ and that $E(\varepsilon_i^2)=1$, and have introduced the notational shorthand
$$
R(\theta):= R_2(\theta) - R_1(\theta)
$$
for
\begin{equation*}
\begin{split}	
R_1(\theta)
& :=
\
P(\lvert \varepsilon_i \rvert > \eta \, \sqrt{m}) \, F(\theta)
\\
& + \phantom{\frac{1}{2}}
E(\lbrace \lvert \varepsilon_i \rvert > \eta \, \sqrt{m} \rbrace \, \varepsilon_i^{} ) \, \frac{\nabla_2^1 G(\theta,0)}{\sqrt{m}}
\\
& + \frac{1}{2}
E(\lbrace \lvert \varepsilon_i \rvert > \eta \, \sqrt{m} \rbrace \, \varepsilon_i^2) \, \frac{\nabla_2^2 G(\theta,0)}{m}
\\ & +
\frac{1}{6}
E(\lbrace \lvert \varepsilon_i \rvert > \eta \, \sqrt{m} \rbrace \, \varepsilon_i^3) \, \frac{\nabla_3^2 G(\theta,0)}{m^{3/2}}
\end{split}
\end{equation*}
and
$$
R_2(\theta):=
\frac{1}{24} \frac{E(\lbrace \lvert \varepsilon_i \rvert \leq \eta \, \sqrt{m} \rbrace \, \varepsilon_i^4 \, \nabla_2^4 G(\theta,\varepsilon_i^*/\sqrt{m}))}{m^2};
$$
here, $\varepsilon_i^*$ lies in between zero and $\varepsilon_i$. To validate our bias expression it remains only to establish that $\sup_\theta \lvert R(\theta) \rvert = O(m^{-3/2})$. To do so we show that 
$
\sup_\theta \lvert R_1(\theta) \rvert = O(m^{-2}),
$
and that
$
\sup_\theta \lvert R_2(\theta) \rvert = O(m^{-2}),
$
in turn.
By H\"older's inequality,
$$
\lvert E(\lbrace \lvert \varepsilon_i \rvert > \eta \, \sqrt{m} \rbrace \, \varepsilon_i^{})
\rvert
\leq 
E(\lbrace \lvert \varepsilon_i \rvert > \eta \, \sqrt{m} \rbrace)^{\nicefrac{3}{4}}
\,
E(\varepsilon_i^4)^{\nicefrac{1}{4}}
=
O(P(\lvert \varepsilon_i\rvert > \eta \, \sqrt{m})^{\nicefrac{3}{4}}) 
=
O(m^{-3/2}).
$$
In the same way,
$$
\lvert E(\lbrace \lvert \varepsilon_i \rvert > \eta \, \sqrt{m} \rbrace \, \varepsilon_i^{2})
=
O(m^{-1}),
\qquad
\lvert E(\lbrace \lvert \varepsilon_i \rvert > \eta \, \sqrt{m} \rbrace \, \varepsilon_i^{3})
=
O(m^{-1/2}),
$$
follow. Consequently,
$$
\sup_\theta \lvert R_1(\theta) \rvert
=
O(m^{-2}) \, 
\sup_\theta (1+\nabla_2^1 G(\theta,0)+\nabla_2^2 G(\theta,0)+\nabla_2^3 G(\theta,0))
=
O(m^{-2}),
$$
using that all relevant derivatives on the right-hand side are bounded. Next, noting that, as $\lvert \varepsilon_i^* \rvert\leq \lvert \varepsilon_i \rvert $, the event $\lvert \varepsilon_i \rvert/ \sqrt{m} \leq \eta  $ implies that $\lvert \varepsilon_i^* \rvert/ \sqrt{m} \leq \eta$, we have
\begin{equation*}
\begin{split}
\sup_\theta \lvert R_2(\theta) \rvert 
& = 
\frac{1}{24} \frac{\sup_\theta E(\lbrace \lvert \varepsilon_i \rvert/ \sqrt{m} \leq \eta  \rbrace \, \varepsilon_i^4 \, \nabla_2^4 G(\theta,\varepsilon_i^*/\sqrt{m}))}{m^2}	
\\
& \leq
\frac{1}{24} 
\frac{\sup_\theta \sup_{\delta\in [-\eta,\eta]} \lvert \nabla_2^4 G(\theta,\delta) \rvert \, E(\varepsilon_i^4)}{m^2}
\\
&= O(m^{-2}),
\end{split}
\end{equation*}		
because of \eqref{DerivativesH3}. Therefore,  $\sup_\theta\lvert R(\theta) \rvert = O(m^{-2})$, and so
$$
E(\hat{F}(\theta)) = F(\theta) + \frac{b_F(\theta)}{m} + O(m^{-3/2}),
$$
as before.

Now turning to the result on the covariance, note that
\begin{equation*}
\begin{split}	
\mathrm{cov}(\hat{F}(\theta_1),\hat{F}(\theta_2))
 =
\frac{E(\hat{F}(\theta_1 \wedge \theta_2)) - E(\hat{F}(\theta_1))  \, E(\hat{F}(\theta_2))}{n} 
\end{split}
\end{equation*}
depends only on $E(\hat{F}(\theta))$ which, up to $O(m^{-3/2})$ and uniformly in $\theta$, has been calculated above. Moreover, 
\begin{equation*}
\begin{split}	
\mathrm{cov}(\hat{F}(\theta_1),\hat{F}(\theta_2))
& = 
\frac{\left(F(\theta_1\wedge \theta_2) + O(m^{-1})\right) -
\left(F(\theta_1) + O(m^{-1})\right)
\left(F(\theta_2) + O(m^{-1})\right)}{n}
\\
& = 
\frac{F(\theta_1 \wedge \theta_2) - F(\theta_1) F(\theta_2)}{n} + O(n^{-1} m^{-1})
\\
& = 
\frac{\sigma_F(\theta_1,\theta_2)}{n} + O(n^{-1} m^{-1}),
\end{split}
\end{equation*}
as stated in the proposition.

To complete the proof it remains only to verify the limit distribution of the scaled empirical distribution function.
Let $F_{m}(\theta):=E(\mathrm{1}\lbrace\vartheta_i \leq \theta \rbrace)$, the distribution function of $\vartheta_i$. Our assumptions imply that $F_m$ is continuous and that it has no mass points. With $u_i:= F_m(\vartheta_i)$, we therefore have that  $u_i$ is i.i.d.~uniformly distributed on $[0,1]$ by the probability integral transform. An application of Lemma \ref{lemma:KMTapprox} with $u= F_m(\theta)$ and exploiting monotonicity of distribution functions then gives
$$
\sup_\theta \left\lvert \sqrt{n} (\hat{F}(\theta)-F_m(\theta)) - \mathbb{B}_n(F_m(\theta))  \right\rvert
=
O_p(\log(n)/\sqrt{n}).
$$
We have already shown that, uniformly in $\theta$, 
$$
F_m(\theta) = F(\theta) + \frac{b_F(\theta)}{m} + O(m^{-3/2}).
$$
Therefore, using that $n/m^3\rightarrow 0$ if $n/m^2\rightarrow c\in [0,+\infty)$ as $n,m\rightarrow\infty$,
$$
\sqrt{n} (\hat{F}(\theta)-F_m(\theta))
=
\sqrt{n} \left(\hat{F}(\theta)-F(\theta) - \frac{b_F(\theta)}{m}\right) + o(1),
$$
holds uniformly in $\theta$. Furthermore, our bias calculation implies that $F_m(\theta)-F(\theta)$ converges to zero uniformly in $\theta$ as $m\rightarrow 0$, so that applying L\'evy's modulus-of-continuity theorem, that is,
      \begin{align*}
          \lim_{\epsilon \rightarrow 0} \;
                \sup_{t \in [0,1-\epsilon]}  
                 \frac{ \left|  \mathbb{B}_n( t  )  -  \mathbb{B}_n( t+ \epsilon  ) \right| }
       {  \sqrt{\epsilon \log(1/\epsilon)} }
              = O(1) , \qquad \epsilon > 0,
      \end{align*}
to our problem yields 
$
\sup_\theta \lvert \mathbb{B}_n(F_m(\theta)) - \mathbb{B}_n(F(\theta)) \rvert \overset{p}{\rightarrow} 0 
$
as $m\rightarrow \infty$. We thus have that $\mathbb{B}_n(F_m(\theta))\rightsquigarrow \mathbb{B}_n(F(\theta))$. Putting everything together and noting that, by definition, $\mathbb{B}_n(F(\theta))= \mathbb{G}_F(\theta)$, we obtain
$$
\sup_\theta \left\lvert 
\sqrt{n} \left(\hat{F}(\theta)-F(\theta)-\frac{b_F(\theta)}{m}\right) 
- 
\mathbb{G}_F(\theta) 
\right\rvert
=
o_p(1),
$$
which completes the proof of the proposition.
\end{proof}

\subsection*{Proof of Proposition~\ref{prop:quantile}}

\begin{lemma}     \label{lemma:helpCDFpreliminaries}
Let Assumptions \ref{ass:basicMOM} and \ref{ass:cdf} hold. Let $f_m$ denote the density function of $\vartheta_i$. Then,

\medskip\noindent (i) \
$\sup_\theta \lvert f_m(\theta)- f(\theta) \rvert = O(m^{-1})$,

\medskip\noindent (ii) \
$\sup_\theta \lvert \nabla^1 f_m(\theta)- \nabla^1 f(\theta) \rvert = O(m^{-1})$,

\medskip\noindent (iii) \
$\sup_\theta \lvert \nabla^2 f_m(\theta) - \nabla^2 f(\theta) \rvert = O(1)$,

\medskip\noindent (iv) \
$\sup_\theta \lvert \nabla^3 f_m(\theta) - \nabla^3 f(\theta) \rvert = O(1)$.
\end{lemma}

\begin{proof}
From the argument in the proof of Proposition~\ref{prop:cdf} we have
$$
F_m(\theta) -  F(\theta) 
=
\frac{1}{2}
\frac{E(\varepsilon_i^2 \, H(\theta,\varepsilon_i^*/\sqrt{m}))}{m}
$$
by a second-order expansion, where $\varepsilon_i^*$ is a value between zero and $\varepsilon_i$ and we introduce the function
$$
H(\theta,\delta) : = 
\textstyle{\int}_{\underline{\sigma}}^{\overline{\sigma}} 
\sigma^2 \nabla_1^1 h(\theta-\delta \sigma \vert \sigma) \, h(\sigma) \, d\sigma,
$$
where $h(\theta_i \vert \sigma_i)$ and $h(\sigma_i)$ are the density functions of $\theta_i$ given $\sigma_i$ and of $\sigma_i$, respectively. Differentiating with respect to $\theta$ yields the first conclusion of the lemma as
$$
\sup_\theta
\lvert
f_m(\theta) - f(\theta) 
\rvert
= 
\sup_\theta
\left\lvert
\frac{1}{2}
\frac{E(\varepsilon_i^2 \,\nabla^1_1 H(\theta,{\varepsilon_i^*/\sqrt{m}}))}{m}
\right\rvert
\leq
\frac{E(\sigma_i^2)}{m} \, \frac{\sup_\theta \sup_\delta \lvert \nabla_1^1 H(\theta,\delta) \rvert}{2}
=
O(m^{-1}),
$$
which follows from the inequality 
$$
\sup_\theta \sup_\delta \lvert \nabla_1^1 H(\theta,\delta) \rvert
= 
\sup_\theta \sup_\delta
\left\lvert
\textstyle{\int}_{\underline{\sigma}}^{\overline{\sigma}} \sigma^{3} \, \nabla_1^{2}\, h(\theta - \delta \sigma \vert \sigma) \, h(\sigma) \, d\sigma
\right\rvert
\leq
\textstyle{\int}_{\underline{\sigma}}^{\overline{\sigma}} \sigma^{3} \, e(\sigma) \, h(\sigma) \, d\sigma < \infty
$$
and the definition of the function $e(\sigma)$ in Assumption \ref{ass:cdf}. The second conclusion of the lemma follows in the same manner, differentiating once more. Finally, the third and fourth conclusion are obtained similarly. The point of departure is now the following identity, which is derived in the proof of Proposition~\ref{prop:cdf},
$$
F_m(\theta)  = E\left(G(\theta,{\varepsilon^*_i/\sqrt{m}})\right)
$$
where 
$$
G(\theta,\delta) : = 
\textstyle{\int}_{\underline{\sigma}}^{\overline{\sigma}} \int_{-\infty}^{\theta-\delta \sigma} h(\vartheta\vert \sigma) \, h(\sigma) \, d\vartheta \, d\sigma .
$$
Repeated differentiation shows that
\begin{equation*}
\begin{split}	
\sup_\theta \sup_\delta\lvert \nabla^3_1 G(\theta,\delta) \rvert
& =
\sup_\theta \sup_\delta\lvert
\textstyle{\int}_{\underline{\sigma}}^{\overline{\sigma}}\nabla_1^2 h(\theta-\delta\sigma\vert \sigma) \, h(\sigma) \, d\sigma
\rvert
\leq
\lvert
\textstyle{\int}_{\underline{\sigma}}^{\overline{\sigma}} e(\sigma) \, h(\sigma) \, d\sigma
\rvert < \infty,
\\
\sup_\theta \sup_\delta\lvert \nabla_1^4 G(\theta,\delta) \rvert
& =
\sup_\theta \sup_\delta\lvert
\textstyle{\int}_{\underline{\sigma}}^{\overline{\sigma}} \nabla_1^3 h(\theta-\delta\sigma\vert \sigma) \, h(\sigma) \, d\sigma
\rvert
\leq
\lvert
\textstyle{\int}_{\underline{\sigma}}^{\overline{\sigma}} e(\sigma) \, h(\sigma) \, d\sigma
\rvert < \infty,
\end{split}
\end{equation*}
and so
$\sup_\theta \lvert \nabla^3 F_m(\theta) \rvert = O(1)$ and $\sup_\theta \lvert \nabla^4 F_m(\theta) \rvert = O(1)$ follow. Furthermore, 
\begin{equation*}
\begin{split}
\sup_\theta \lvert \nabla^2 f_m(\theta) - \nabla^2 f(\theta) \rvert
& \leq
\sup_\theta \lvert \nabla^2 f_m(\theta) \rvert
+
\sup_\theta \lvert \nabla^2 f(\theta) \rvert
=
O(1),
\\
\sup_\theta \lvert \nabla^3 f_m(\theta) - \nabla^3 f(\theta) \rvert
& \leq
\sup_\theta \lvert \nabla^3 f_m(\theta) \rvert
+
\sup_\theta \lvert \nabla^3 f(\theta) \rvert
=
O(1) ,
\end{split}
\end{equation*}
follows because $f$ has uniformly bounded derivatives up to third order by assumption. This completes the proof.
\end{proof}

\begin{proof}[\bf Proof of Proposition~\ref{prop:quantile}]
The $\vartheta_i$ are i.i.d.~draws from the distribution $F_m$ which according to Lemma~\ref{lemma:helpCDFpreliminaries} has non-degenerate density  $f_m$, that is, the $\vartheta_i$ are continuously distributed. Thus, 
$$
u_{(k)} := F_m( \vartheta_{(k)})
$$ 
is the $k$th order statistic of a uniform sample. We set $k=\lceil \tau n \rceil$ for the rest of the proof. Then $\hat{q}(\tau) = \vartheta_{(k)}$. Since $k/n \rightarrow \tau$ by construction, it is well-known that 
   \begin{align}
   \sqrt{n}(u_{(k)} - \tau) \overset{d}{\rightarrow} { N}(0, \tau (1-\tau)) .
       \label{LimitBeta}
   \end{align}
Let $q_m(\tau):=F_m^{-1}(\tau)$, the $\tau$th-quantile of $F_m$. By expanding the function $F_m^{-1}$ around $\tau$ we find that
$$
\hat{q}(\tau) = 
F_m^{-1}( u_{(k)}  )
       =q_m(\tau) + \frac{u_{(k)} - \tau }{f_m(q_m(\tau))}
          + r_{(k)}
$$
for remainder term
$$
r_{(k)} :=- \frac{f_m'(\xi_{(k)})} {f_m(\xi_{(k)})^3} \left(u_{(k)} - \tau \right)^2,
$$
where $\xi_{(k)}$ is a value between  $F_m^{-1}(\tau)$ and $F_m^{-1}( u_{(k)} )$. From \eqref{LimitBeta} we have $u_{(k)} - \tau = O_P(n^{-1/2})$. This implies that $ \xi_{(k)}  \overset{p}{\rightarrow} \tau$. Using Lemma~\ref{lemma:helpCDFpreliminaries} we may conclude that $f_m(\xi_{(k)}) \overset{p}{\rightarrow} f_m(\tau)  \rightarrow f(\tau) >0$, and, therefore, that $ r_{(k)}  = O_p(n^{-1})$.
   We thus have
    \begin{align*}
        \hat{q}(\tau)  
        &= q_m(\tau) + \frac{ u_{(k)} - \tau }{f_m(q_m(\tau))} + O_p(n^{-1}).
    \end{align*}
    Again using Lemma~\ref{lemma:helpCDFpreliminaries} and
    our assumption that $f(\theta)>0$ in a neighborhood of $q(\tau) = F^{-1}(\tau)$
    we have ${f_m(q_m(\tau))^{-1}} = {f(q(\tau))^{-1}} + O(m^{-1})$,
    and therefore 
    \begin{align}
        \hat{q}(\tau)
        &= q_m(\tau) +  \frac{u_{(k)} - \tau }{f(q(\tau))}   + O_p(n^{-1} + n^{-1/2} m^{-1}).
        \label{Qexpansion1}
    \end{align}    
    From Proposition~\ref{prop:cdf} we know $F_m(\theta) = E(\hat{F}(\theta))
         =  F(\theta) +   b_F(\theta)/m  + O(m^{-3/2})$,
     and therefore
   \begin{align}
         q_m(\tau) =   q(\tau) -  \frac{b_F(q(\tau))/f(q(\tau))} {m} + O(m^{-3/2}) .
        \label{Qexpansion2}
   \end{align}      
   Combining  \eqref{LimitBeta}, \eqref{Qexpansion1}, and \eqref{Qexpansion2}
    gives the statement of the theorem.
\end{proof}

\subsection*{Proof of Proposition \ref{prop:WeakConvCorrected}}

\begin{lemma}     \label{lemma:helpCDFpreliminaries4}
Let the assumptions of Proposition \ref{prop:WeakConvCorrected} hold. Then, 
\begin{itemize}
\item[(i)]
$ \displaystyle
\sup_\theta E(\hat{b}_F(\theta)-b_F(\theta)) = O(m^{-1}) + O(h^2) ,$
\medskip
\item[(ii)]
$ \displaystyle
\sup_\theta \mathrm{var}(\hat{b}_F(\theta)) = O(n^{-1}h^3),
$
\medskip
\item[(iii)] %
$ \displaystyle
\sup_\theta \, (1+\lvert\theta \rvert^{1+\eta}) \, \lvert \nabla^1\hat{b}_F(\theta)-\nabla^1 b_F(\theta) \rvert  = O_p(h^{-(\omega+1)/\omega}) .
$	
\end{itemize}
\end{lemma}

\begin{lemma}     \label{lemma:helpCDFpreliminaries5}
Let Assumptions \ref{ass:basicMOM} hold and define
$$
b_i(\theta) : = -\frac{\sigma_i^2}{h^{2}}\frac{ \phi^\prime\left(\frac{\vartheta_i-\theta}{h}\right)}{2} .
$$
If $f$ is bounded, then, for any $\epsilon>0$,
$$
\sup_\theta E(\lvert b_i(\theta) - E(b_i(\theta)) \rvert^\epsilon)^{1/\epsilon}
=
O(h^{-2+\epsilon^{-1}}).
$$

\end{lemma}

\noindent
The proof of those two lemmas is provided below, after the proof of the main text results.

\begin{proof}[\bf Proof of Proposition~\ref{prop:WeakConvCorrected}]
We first show that
\begin{align*}
     \sup_{\theta \in \mathbb{R}}
         \left|  \hat{b}_F(\theta)  -  b_F(\theta) \right|
     &=   O( m^{-1} ) + O(h^2)    + O( n^{-1/2} \, h^{-3/2-\varepsilon} )  .
\end{align*}	
The result of the proposition then follows readily.	
For a finite $\nu$, introduce the function 
$$
t(\theta) :=  {\rm sgn}(\theta) \, \frac{1 - (1+|\theta|)^{-\nu} }{\nu}.
$$
Note that $t$ maps to the finite interval $(-\nu^{-1},\nu^{-1})$ and is monotone increasing; moreover, $\nabla^1 t(\theta) = (1+|\theta|)^{-(1+\nu)}$. Now consider the reparametrization $\tau = t(\theta)$; note that $\tau$ lives in a bounded interval. From Lemma \ref{lemma:helpCDFpreliminaries4}(iii), using the chain rule of differentiation, it follows that
\begin{equation} \label{LipschitzBound}  
\sup_{\tau\in (-\nu^{-1},\nu^{-1})}	\left\lvert \nabla^1_\tau \hat{b}_F(t^{-1}(\tau)) -  \nabla^1_\tau b_F(t^{-1}(\tau)) \right\rvert = O_p(h^{-(1+\omega^{-1})}),
\end{equation}	
where we use the notation $\nabla_\tau$ to indicate derivatives with respect to $\tau$. We therefore have that $\hat{b}_F(t^{-1}(\tau)) - b_F(t^{-1}(\tau))$, as a function $\tau$, has a uniformly-bounded Lipschitz constant. Now let $I_h$ be a partition of $(-\nu,-\nu^{-1})$ with subintervals that are (approximately) of length $l_h:=h^{3-\omega^{-1}}$. Then \eqref{LipschitzBound} implies that
$$
\sup_\theta \lvert \hat{b}_F(\theta)-b_F(\theta) \rvert
=
 \sup_{\tau \in (-\nu, \nu)} \lvert \hat{b}_F(t^{-1}(\tau))-b_F(t^{-1}(\tau)) \rvert
$$
is equal to
\begin{equation} \label{SupInequ1}
\max_{\tau\in I_h} \lvert \hat{b}_F(t^{-1}(\tau))-b_F(t^{-1}(\tau)) \rvert + O_p(h^2).
\end{equation}
Here, the order of the remainder terms follows from the choice of $l_h$. Now introduce the shorthand
$$
\hat{\Delta}(\theta) : = \hat{b}_F(\theta) - E(\hat{b}_F(\theta)).
$$
Then 
$$
\max_{\tau\in I_h} \lvert \hat{b}_F(t^{-1}(\tau))-b_F(t^{-1}(\tau)) \rvert
\leq
\max_{\tau\in I_h} \lvert \hat{\Delta}(t^{-1}(\tau)) \rvert
+
\sup_\theta \lvert E(\hat{b}_F(\theta)) - b_F(\theta) \rvert
$$
and so  Lemma \ref{lemma:helpCDFpreliminaries4}(i) implies that
$$
\max_{\tau\in I_h} \lvert \hat{b}_F(t^{-1}(\tau))-b_F(t^{-1}(\tau)) \rvert
\leq
\max_{\tau\in I_h} \lvert \hat{\Delta}(t^{-1}(\tau)) \rvert + O(m^{-1}+h^2).
$$
Moving on, observe that the number of subintervals making up $I_h$ is  equal to
$\lceil l_h^{-1} \rceil = \lceil h^{-3+\omega^{-1}}  \rceil $,
where $\lceil a \rceil $ delivers the smallest integer at least as large as $a$. We therefore have
\begin{equation} \label{SupInequ2}
\begin{split}	
   E
    \left(\left( \max_{\tau \in I_h}
         \left|  \hat \Delta(t^{-1}(\tau)) \right|   
   \right)^{\omega} \right)     
      &=   
  E  \left( \max_{\tau \in I_h}
         \left|  \hat \Delta(t^{-1}(\tau)) \right|^{\omega}  \right)  
 \\ & \leq   E  \left( \sum_{\tau \in I_h}
         \left|  \hat \Delta(t^{-1}(\tau)) \right|^{\omega} \right)    
 \\
  &=  \sum_{\tau \in I_h}  E  \left( 
         \left|  \hat \Delta(t^{-1}(\tau)) \right|^{\omega}  \right)
		 \leq    \left\lceil h^{-3+1/\omega}  \right\rceil 
       \sup_{\theta \in \mathbb{R}}    
          E \left|  \hat \Delta(\theta) \right|^{\omega}   .      
\end{split}
\end{equation}
Let $b_i(\theta) :=  - \frac 1 2 \, h^{-2} \, \sigma_i^2  \, \phi^\prime\left(\frac{\vartheta_i-\theta}{h}\right) $
 and    $\Delta_i(\theta) :=  b_i(\theta)  - E  b_i(\theta)$.
 We may then write
$\hat \Delta(\theta)   =   n^{-1}\sum_{i=1}^n \Delta_i(\theta)$.
Notice that $\Delta_i(\theta)$ are independent and mean zero. 
By \citet[Theorem 3]{Rosenthal1970}
we therefore have that
$$
  \left(    E \left( \left|  n^{-1/2} \sum_{i=1}^n \Delta_i(\theta) \right|^{\omega} \right) \right)^{1/\omega}
$$
is bounded from above by 
\begin{align*}
 c \, \max \left\{ 
        \left(  n^{-1} \sum_{i=1}^n E \left(\Delta_i(\theta)^2 \right) \right)^{1/2},
          \;
       n^{-1/2}  \left( \sum_{i=1}^n E  \left(\left| \Delta_i(\theta) \right|^{\omega}\right) \right)^{1/\omega}
      \right\} ,
\end{align*}
where the constant $c$ only depends on $\omega$.
Using Lemma~\ref{lemma:helpCDFpreliminaries4}(ii)  we  obtain
\begin{align*}
       \sup_{\theta \in \mathbb{R}}  \,     \left( n^{-1} \sum_{i=1}^n E   (\Delta_i(\theta)^2 )\right)^{1/2}
     &= \sup_{\theta \in \mathbb{R}}  
       \left(   n  \, {\rm var} \, \hat{b}_F(\theta) \right)^{1/2} = O(h^{-3/2}).
\end{align*}
Using Lemma \ref{lemma:helpCDFpreliminaries5} we obtain
\begin{align*}
 n^{-1/2} \, \sup_{\theta \in \mathbb{R}} 
     \left( \sum_{i=1}^n E  \left( \left| \Delta_i(\theta) \right|^{\omega} \right)^{1/\omega}\right)
   &=  n^{-1/2+1/\omega}  \; \sup_{\theta \in \mathbb{R}}     \left(    E   \left| \Delta_i(\theta) \right|^{\omega} \right)^{1/\omega} 
 \\  
     &=    O( n^{-1/2+1/\omega} \,  h^{-2+1/\omega} ) 
=  O( h^{-3/2} ) ,   
\end{align*}
where in the last step we used the condition that $h^{-1} = O(n)$.
We can therefore conclude from Rosenthal's inequality above that
\begin{align*}
    \left(   \sup_{\theta \in \mathbb{R}}    
          E \left( |  \hat \Delta(\theta) |^{\omega} \right)
    \right)^{1/\omega}      
    &=
     n^{-1/2} \left(    E \left( \left|  n^{-1/2} \sum_{i=1}^n \Delta_i(\theta) \right|^{\omega} \right) \right)^{1/\omega}
    = O( n^{-1/2} h^{-3/2} ) .
\end{align*}
Using  this and  \eqref{SupInequ2} we obtain
\begin{align*}
     \max_{\tau \in I_h}
         \left|  \hat \Delta(t^{-1}(\tau)) \right|  
      &=    O(h^{(-3+1/\omega)/\omega}  \, n^{-1/2} \, h^{-3/2} )  
       =    O( n^{-1/2} \, h^{-3/2-\varepsilon} )  ,
\end{align*}
where $\varepsilon = 3/\omega - 1/\omega^2$.
Combining this with \eqref{SupInequ1} and \eqref{SupInequ2}  we thus conclude
\begin{align*}
     \sup_{\theta \in \mathbb{R}}
         \left|  \hat{b}_F(\theta)  -  b_F(\theta) \right|
     &=   O( m^{-1} ) + O(h^2)    + O( n^{-1/2} \, h^{-3/2-\varepsilon} )   ,
\end{align*}
as claimed.

Now, with $h = O(m^{-1/2})$ and $h^{-1} = O(n^{1 - 2 \omega^{-1}})$ we find
\begin{align*}
       \sup_{\theta \in \mathbb{R}}
        \frac{\sqrt{n}} {m}  \left|  \hat{b}_F(\theta)  -  b_F(\theta) \right| 
         &= O_P( n^{1/2} m^{-1} h^2 + n^{1/2} m^{-2} +   m^{-1} h^{-3/2-\varepsilon} )
      \\   
        &= O_P( n^{1/2} m^{-2} + m^{-4/9 \epsilon^2})
      \\ 
       &= o_P(1) ,
\end{align*}
where in the last step we also used that $n/m^4 \rightarrow 0$ and that $m \rightarrow \infty$. The result of Proposition~\ref{prop:WeakConvCorrected} now follows immediately from Proposition~\ref{prop:cdf}.
\end{proof}

\subsection*{Proof of Proposition~\ref{prop:Qbias}}

Let $\mathbb{G}_n(u) := \hat F( F_m^{-1}(u)  )$ be the empirical distribution function of the i.i.d.~sample $u_i = F_m(\vartheta_i)$.  Lemma~\ref{lemma:KMTapprox} and Theorem~1 in \cite{DossGill1992} give
\begin{align}
\sup_{\tau \in [0, 1]} \left|   
\sqrt{n}  \left( \mathbb{G}^{\leftarrow}_n(\tau)  - \tau \right) + \mathbb{B}_n(\tau)   \right| = o_P(1) ,
\label{QprocessApprox}
\end{align}
where $\mathbb{G}^{\leftarrow}_n$ again denotes the left inverse of $\mathbb{G}_n$ $\mathbb{B}_n(\tau)$ is the sequence of Brownian bridges that previously appeared in Lemma~\ref{lemma:KMTapprox}.	

Equation \eqref{QprocessApprox} yields
\begin{align*}
\mathbb{G}_n^{\leftarrow}( \hat{\tau}^*)  - \mathbb{G}_n^{\leftarrow}( \tau ) = ( \hat{\tau}^* - \tau  ) - n^{-1/2} \left[ \mathbb{B}_n(\hat{\tau}^*) - \mathbb{B}_n(\tau)  \right] + o_p(n^{-1/2}).
\end{align*}   
Also, $ \hat{\tau}^* - \tau  = O_p( m^{-1} )$ follows from the results above. L\'evy's modulus-of-continuity theorem then implies that $\mathbb{B}_n(\hat{\tau}^*) - \mathbb{B}_n(\tau) = o_P(1)$. Therefore, 
\begin{align*}
\mathbb{G}_n^{\leftarrow}( \hat{\tau}^*)  - \mathbb{G}_n^{\leftarrow}( \tau ) = O_p(m^{-1} ) + o_p(n^{-1/2})  .
\end{align*} 	
By definition we have $  \check{q}(\tau) = \hat F^{\leftarrow}(\hat{\tau}^*)$ and $ \hat{q}(\tau) = \hat F^{\leftarrow}( \tau  )$, and also that $\mathbb{G}^{\leftarrow}_n(\tau) =  F_m( \hat F^{\leftarrow}(\tau)  )$. Substituting this into the last displayed equation yields	
\begin{align*}
F_m( \check{q}(\tau) ) - F_m( \hat{q}(\tau)) = O_p(m^{-1} ) + o_p(n^{-1/2})  .
\end{align*}	
Lemma~\ref{lemma:helpCDFpreliminaries} and our assumptions guarantee that $F_m(\tau)$ has a density $f_m(\tau)$ that is bounded from below in a neighborhood of $q(\tau)$ for the quantile of interest $\tau$. The last result therefore also implies that
\begin{align} \label{DiffEstQ}
\check{q}(\tau) -  \hat{q}(\tau) = O_p(m^{-1} ) + o_p(n^{-1/2})  .
\end{align}

Next, 
The result \eqref{QprocessApprox} implies $\sqrt{n} (\mathbb{G}_n^{\leftarrow}(\tau) - \tau)  \rightsquigarrow   \mathbb{B}(\tau)$ for a Brownian bridge $ \mathbb{B}$. For $\check{q}(\tau) = \hat F^{\leftarrow}(\hat{\tau}^*)$ we have $F_m(\check{q}(\tau)) = \mathbb{G}^{\leftarrow}_n(\hat{\tau}^*) $, and therefore 
$$
\sqrt{n} (F_m(\check{q}(\tau)) - \hat{\tau}^*)  \rightsquigarrow   \mathbb{B}(\tau).
$$
From Proposition~\ref{prop:cdf} we know that $F_m(\theta) = E(\hat{F}(\theta))  = F(\theta) +  {b_F(\theta)}/{m} + O(m^{-2})$, uniformly in $\theta$. We then find
\begin{align*}
\sqrt{n} \left(  F(\check{q}(\tau))  -\tau +  \frac{b_F(\check{q}(\tau)) - \hat{b}_F(\hat{q}(\tau)) }{m}   + O(m^{-2}) \right) 
\overset{d}{\rightarrow}{ N}(0, \tau (1-\tau)) ,
\end{align*}
From the proof of Proposition~\ref{prop:WeakConvCorrected} we also know that   
$
\sup_{\theta}
({\sqrt{n}}/ {m})  \left\lvert  \hat{b}_F(\theta)  -  b_F(\theta)  \right\rvert 
= o_p(1) ,
$
and therefore 
    \begin{align*}
          \sqrt{n} \left(  F(\check{q}(\tau))  -\tau
          +  \frac{b_F(\check{q}(\tau)) -  b_F(\hat{q}(\tau)) }{m}   + O(m^{-2})  \right)  \overset{d}{\rightarrow} { N}(0, \tau (1-\tau)).
    \end{align*}    
Smoothness of the function  $b_F$ and \eqref{DiffEstQ} imply
$
b_F(\check{q}(\tau)) -  b_F(\hat{q}(\tau)) = O(m^{-1}) + o_p(n^{-1/2}).
$ 
We thus obtain
$
 \sqrt{n} \left(  F(\check{q}(\tau))  -\tau \right)  \overset{d}{\rightarrow} { N}(0, \tau (1-\tau))
$
An application of the delta method with transformation $F^{-1}$ then gives the result. This completes the proof.
\qed %

\renewcommand{\theequation}{B.\arabic{equation}} 
\setcounter{equation}{0}
\renewcommand{\thelemma}{B.\arabic{lemma}} 
\setcounter{lemma}{0}
\section{Proof of Lemmas~\ref{lemma:helpCDFpreliminaries4} and \ref{lemma:helpCDFpreliminaries5}}

Before proving Lemmas~\ref{lemma:helpCDFpreliminaries4} and \ref{lemma:helpCDFpreliminaries5} we first 
state one known result and establish two further intermediate lemmas.

\begin{lemma}[\citealt{Mason1981}] \label{lemma:BoundWeightedEmp}
Let $\mathbb{G}_n$ be the empirical cumulative distribution of an i.i.d. sample of size n from a uniform distribution on [0,1].
Then, as $n\rightarrow \infty$,  	
\begin{align*}
    \sup_{u \in (0,1)} \;    [u (1-u)]^{-1+\epsilon} \, \left| \mathbb{G}_n(u) - u \right|  \, \rightarrow \, 0  ,
\end{align*}
almost surely, for any $0 < \epsilon \leq 1/2$.
\end{lemma}

\begin{lemma}     \label{lemma:helpCDFpreliminaries2}
Let Assumption \ref{ass:basicMOM}  hold. Then, if $\sup_\theta (1+\lvert \theta \rvert^\kappa) \, f(\theta) < \infty$, 
$$
\sup_\theta (1+\lvert \theta \rvert^\kappa) \, f_m(\theta) =  O_p(1).
$$
holds.
\end{lemma}

\begin{proof}
The conditional density of $\vartheta_i-\theta_i$ given $\theta_i$ evaluated in $\varepsilon$ is
$$
p(\varepsilon \vert \, \theta) :=
E\left(\left. \frac{1}{\sigma_i/\sqrt{m}} \, g\left(\frac{\varepsilon}{\sigma_i/\sqrt{m}}\right)\right\rvert \theta_i=\theta\right).
$$	
We thus have
$$
f_m(\vartheta) 
= 
\textstyle{\int}_{-\infty}^{\infty} p(\vartheta-\theta \vert \, \theta) \, f(\theta) \, d\theta
= 
\textstyle{\int}_{-\infty}^{\vartheta/2} \, p(\vartheta-\theta \vert \, \theta) \, f(\theta) \, d\theta
+
\textstyle{\int}_{\vartheta/2}^{\infty} \, p(\vartheta-\theta \vert \, \theta) \, f(\theta) \, d\theta .
$$
Without loss of generality we will take the value $\vartheta$ to be positive throughout. We have the bound
\begin{equation} \label{eq:fbound}
f_m(\vartheta) 
\leq 
\sup_{\theta} f(\theta) \, \textstyle{\int}_{-\infty}^{\vartheta/2} \, p(\vartheta-\theta \vert \, \theta) \, d\theta
+
\sup_{\theta\geq \vartheta/2} f(\theta) \,
\textstyle{\int}_{\vartheta/2}^{\infty} \, p(\vartheta-\theta \vert \, \theta) \, d\theta .
\end{equation}
Consider the second term on the right-hand side in \eqref{eq:fbound}. $\sup_{\theta \geq \vartheta/2} f(\theta) = O(1+\lvert\vartheta/2 \rvert^{-\kappa})$ by assumption and so it suffices to show that the integral is finite for all $\vartheta$. To see that this is so, observe that
\begin{equation*}
\begin{split}	
\textstyle{\int}_{\vartheta/2}^{\infty} \, p(\vartheta-\theta \vert \, \theta) \, d\theta
& = 
\textstyle{\int}_{-\infty}^{\vartheta/2} \, p(\varepsilon \vert \, \vartheta-\varepsilon) \, d\varepsilon
=
\textstyle{\int}_{-\infty}^{\vartheta/2}
E\left(\left.
\frac{1}{\sigma_i/\sqrt{m}} \, g\left(\frac{\varepsilon}{\sigma_i/\sqrt{m}}\right)
\right\rvert \theta_i = \vartheta - \varepsilon\right) \, d\varepsilon
\end{split}
\end{equation*}
and use the change of variable $\epsilon^* = \sqrt{m}\, \varepsilon$
\begin{equation*}
\begin{split}
\textstyle{\int}_{\vartheta/2}^{\infty} \, p(\vartheta-\theta \vert \, \theta) \, d\theta
& \leq 
\textstyle{\int}_{-\infty}^{\infty}	
\max_{\sigma\in [\underline{\sigma},\overline{\sigma}]}
\left\lbrace
\frac{1}{\sigma/\sqrt{m}} \, g\left(\frac{\varepsilon}{\sigma/\sqrt{m}}\right)
\right\rbrace \, d\varepsilon
 =
\textstyle{\int}_{-\infty}^{\infty}	
\max_{\sigma\in [\underline{\sigma},\overline{\sigma}]}
\left\lbrace
\frac{1}{\sigma} \, g\left(\frac{\epsilon}{\sigma}\right)
\right\rbrace \, d\epsilon
\\
& \leq 
C \, 
\textstyle{\int}_{-\infty}^{\infty}	
\max_{\sigma\in [\underline{\sigma},\overline{\sigma}]}
\left\lbrace
\frac{1}{\sigma} \, \left(1+\left\lvert \frac{\epsilon}{\sigma}\right\rvert\right)^{-\alpha}
\right\rbrace \, d\epsilon
\\
& \leq 
C \, 
\textstyle{\int}_{-\infty}^{\infty}	
\frac{1}{\underline{\sigma}} \, \left(1+\left\lvert \frac{\epsilon}{\underline{\sigma}}\right\rvert\right)^{-\alpha} \, d\epsilon
 = 
C / (\alpha-1) = O(1).
\end{split}
\end{equation*}		
Next, for the first right-hand side term in \eqref{eq:fbound}, recall that $\sup_\theta f(\theta) < \infty$, and so we need to show that the integral vanishes sufficiently fast as $\vartheta\rightarrow \infty$. To see that this is the case we proceed as before by observing that
\begin{equation*}
\begin{split}	
\textstyle{\int}_{-\infty}^{\vartheta/2} \, p(\vartheta-\theta \vert \, \theta) \, d\theta
& = 
\textstyle{\int}_{\vartheta/2}^{\infty}
E\left(\left.
\frac{1}{\sigma_i/\sqrt{m}} \, g\left(\frac{\varepsilon}{\sigma_i/\sqrt{m}}\right)
\right\rvert \theta_i = \vartheta - \varepsilon\right) \, d\varepsilon
\end{split}
\end{equation*}
to obtain 
\begin{equation*}
\begin{split}
\textstyle{\int}_{-\infty}^{\vartheta/2} \, p(\vartheta-\theta \vert \, \theta) \, d\theta
& \leq 
\textstyle{\int}_{\vartheta/2}^{\infty} \, 
\max_{\sigma\in[\underline{\sigma},\overline{\sigma}]}
\left\lbrace
\frac{1}{\sigma/\sqrt{m}}
\,g \left(\frac{\varepsilon}{\sigma/\sqrt{m}}\right)
\right\rbrace
\, d\varepsilon
\\
& \leq 
\textstyle{\int}_{\sqrt{m}\, \vartheta/2}^{\infty} \, 
\max_{\sigma\in[\underline{\sigma},\overline{\sigma}]}
\left\lbrace
\frac{1}{\sigma}
\,g \left(\frac{\epsilon}{\sigma}\right)
\right\rbrace
\, d\epsilon
\\
& \leq 
C \, 
\textstyle{\int}_{\sqrt{m}\, \vartheta/2}^{\infty} \, 
\frac{1}{\underline{\sigma}}
\left(1+\frac{\epsilon}{\underline{\sigma}} \right)^{-\alpha}
\, d\epsilon
 = 
O(1+(\sqrt{m}\, \vartheta/2)^{1-\alpha}).
\end{split}
\end{equation*}		
Thus, as long as $\alpha>1$ and $\alpha\geq \kappa+1$ we have
$$
f_m(\vartheta) = O(1+\lvert\vartheta/2 \rvert^{-\kappa})
$$
uniformly in $\vartheta$, as claimed. This completes the proof of the lemma.	
\end{proof}

\begin{lemma}     \label{lemma:helpCDFpreliminaries3}
Let Assumptions \ref{ass:basicMOM} and \ref{ass:cdf} hold and let
$$
\gamma_m^r(\theta) : = E(\sigma_i^r \vert \vartheta_i=\theta) \, f_m(\theta),
\qquad
\gamma^r(\theta) : = E(\sigma_i^r \vert \theta_i=\theta) \, f(\theta).
$$
Then, for any integer $r$,
$$
\sup_\theta \lvert \nabla^q\gamma_m^r(\theta)-\nabla^q \gamma^r(\theta) \rvert = O(m^{-1})
$$
provided that the conditional density $h(\theta \vert \sigma)$ is $(q+2)$ times differentiable with respect to $\theta$ and that there exists a function $e$ so that $\lvert \nabla_1^{q+2} h(\theta \vert \sigma) \rvert \leq e(\sigma)$ and $E(e(\sigma_i))<\infty$.
\end{lemma}

\begin{proof}
Fix $r$ throughout the proof. First note that, by Bayes' rule and Assumption \ref{ass:basicMOM}, we may write
\begin{equation*}
\begin{split}
\gamma_m^r(\vartheta) = 	
\textstyle{\int}_{\underline{\sigma}}^{\overline{\sigma}}
\textstyle{\int}_{-\infty}^{\infty} 
\sigma^r \frac{1}{\sigma/\sqrt{m}} \, g\left(\frac{\vartheta-\theta}{\sigma/\sqrt{m}}\right) \, h(\theta, \sigma)  \, d\sigma \, d\theta 	
\end{split}
\end{equation*}	
A change of variable from $\theta$ to $\varepsilon:=(\vartheta-\theta)/(\sigma/\sqrt{m})$	then allows to write
$$
\gamma_m^r(\vartheta) 
=
E\left(B_r(\vartheta,\varepsilon_i/\sqrt{m}) \right),
\qquad
B_r(\theta,\delta) : = 
\textstyle{\int}_{\underline{\sigma}}^{\overline{\sigma}} \sigma^r \, h(\theta-\delta\sigma,\sigma) \, d\sigma.
$$
Observe that $B_r(\vartheta,0) = \gamma^r(\vartheta)$. Now, by a Taylor expansion,
$$
\nabla^q \gamma_m^r(\vartheta) - \nabla^q \gamma^r(\vartheta) = 
\frac{E\left({\varepsilon_i^2} \, \nabla_1^q \nabla_2^{2} B_r(\vartheta,\varepsilon_i^*/\sqrt{m}) \right)}{m}.
$$
Also, as
$$
\nabla_{1}^{p}\nabla_{2}^{q} B_r(\theta,\delta)
=
(-1)^q \, 
\textstyle{\int}_{\underline{\sigma}}^{\overline{\sigma}}
\sigma^{r+q} \, \nabla_1^{p+q} h(\theta-\delta \sigma, \sigma) \, d\sigma
$$
for any pair of integers $(p,q)$, we have that
$$
\sup_\theta \sup_\delta 
\lvert \nabla_1^q \nabla_2^{2} B_r(\theta,\delta)  \rvert
\leq
\overline{\sigma}^{r+q} \,  \, \sup_\theta \sup_\delta  
\lvert \textstyle{\int}_{\underline{\sigma}}^{\overline{\sigma}}
\nabla_1^{2+q} h(\theta-\delta \sigma \vert \sigma) \, h(\sigma) \, d\sigma  \rvert
\leq
\overline{\sigma}^{r+q} \,  
 \textstyle{\int}_{\underline{\sigma}}^{\overline{\sigma}} e(\sigma) \, h(\sigma) \, d\sigma  ,
$$
which is finite. Therefore, uniformly in $\theta$,
$$
\nabla^q \gamma_m^r(\theta) - \nabla^q \gamma^r(\theta) = 
 O(m^{-1}),
$$
as claimed. This completes the proof.
\end{proof}

\begin{proof}[\bf Proof of Lemma~\ref{lemma:helpCDFpreliminaries4}]
	
\mbox{}

\medskip \noindent	
\underline{Part (i)}:
With
$$
\beta_m(\theta) : = \frac{E(\sigma_i^2 \vert \vartheta_i = \theta) \, f_m(\theta)}{2},
$$	
a change of variable and integration by parts yield
$$
E(\hat{b}_F(\theta)) = 
- \textstyle{\int}_{-\infty}^{\infty}
\frac{\beta_m(\vartheta)}{h^2} \, \phi^\prime\left(\frac{\vartheta-\theta}{h} \right) \, d\vartheta
=
\textstyle{\int}_{-\infty}^{\infty}
\nabla^1 \beta_m (\theta+h\varepsilon) \, \phi(\varepsilon) \, d\varepsilon.
$$
Taylor expanding $\nabla^1 \beta_m$ around $\varepsilon = 0$ and using our assumptions of the distribution of $\varepsilon$ we obtain
$$
E(\hat{b}_F(\theta)) = 
\nabla^1 \beta_m(\theta) + 
h^2 \, \frac{\textstyle{\int}_{-\infty}^{\infty} \nabla^3\beta_m(\theta+h \varepsilon^*) \, \varepsilon^2 \, \phi(\varepsilon) \, d\varepsilon}{2},
$$
where $\varepsilon^*$ lies between $\varepsilon$ and zero. From Lemma \ref{lemma:helpCDFpreliminaries3} we have
$$
\nabla^1\beta_m(\theta) = \nabla^1\beta(\theta) + O(m^{-1}) = b_F(\theta) + O(m^{-1}) ,
$$
uniformly in $\theta$, and $\sup_\theta \lvert \nabla^3\beta_m(\theta) \rvert <\infty$. Therefore,
$$
E(\hat{b}_F(\theta)) = b_F(\theta) + O(m^{-1}) + O(h^2),
$$
as claimed.

\bigskip\noindent
\underline{Part (ii)}:
Note that
$$
\mathrm{var}(\hat{b}_F(\theta)) = 
E(\hat{b}_F(\theta)^2) - E(\hat{b}_F(\theta))^2
=
\frac{n^{-1}}{4}
{E\left(\frac{\sigma_i^4}{h^4} \phi^\prime\left(\frac{\vartheta-\theta}{h}\right)^2 \right)}
-
b_F(\theta)^2 + o(n^{-1}). %
$$
Now, with
$$
\beta_m^2(\theta) : = \frac{E(\sigma_i^4 \vert \vartheta_i = \theta) \, f_m(\theta)}{4},
$$
we have
$$
\frac{n^{-1}}{4}
E\left(\frac{\sigma_i^4}{h^4} \phi^\prime\left(\frac{\vartheta-\theta}{h}\right)^2 \right)
=
{\textstyle{\int}_{-\infty}^{\infty} \frac{\beta_m^2(\vartheta)}{h^4} \phi^\prime\left(\frac{\vartheta-\theta}{h}\right)^2 \, d\vartheta}
\leq
\frac{\sup_\theta \lvert \beta_m^2(\theta)\rvert}{n} \,
\frac{\textstyle{\int}_{-\infty}^{\infty} \phi^\prime\left(\frac{\vartheta-\theta}{h}\right)^2 \, d\vartheta}{h^4}
$$
which is $O(n^{-1}h^3)$ uniformly in $\theta$ as $\sup_\theta \lvert \beta_m^2(\theta)\rvert<\infty$ because $\sigma_i$ is finite and $f_m$ is bounded, and
$$
\textstyle{\int}_{-\infty}^{\infty} \phi^\prime\left(\frac{\vartheta-\theta}{h}\right)^2 \, d\vartheta
= \frac{h}{4\sqrt{\pi}},
$$
independent of $\theta$. This completes the proof.	

\bigskip\noindent
\underline{Part (iii)}:
First observe that
$$
\nabla^1 b_F(\theta) = {\nabla^2\beta (\theta)}/2,
$$
so that $(1+\lvert \theta \rvert^{1+\eta}) \, \lvert \nabla^1 b_F(\theta) \rvert < \infty$ follows directly from Assumption~\ref{ass:cdfcorrect}. What is left to show is that 
$$
\sup_\theta (1+\lvert \theta \rvert^{1+\eta}) \, \lvert \nabla^1 \hat{b}_F(\theta) \rvert = O_p(-(1+\omega^{-1})).
$$
Note that
$$
\nabla^1 \hat{b}_F(\theta) = \frac{(nh^2)^{-1}}{2} \sum_{i=1}^n \sigma_i^2 \, \phi^{\prime\prime} \left(\frac{\vartheta_i-\theta}{h}\right).
$$
By H\"older's inequality, 
$$
\lvert \nabla^1 \hat{b}_F(\theta) \rvert \leq
h^{-2} \,
\left \lbrace 
\left( n^{-1} \sum_{i=1}^n (\sigma_i^2/2)^{\omega} \right)^{\omega^{-1}}
\right\rbrace
\times
\left \lbrace 
\left( 
n^{-1} \sum_{i=1}^n \left\lvert \phi^{\prime\prime} \left(\frac{\vartheta_i-\theta}{h}\right)  \right\rvert^\psi
\right)^{\psi^{-1}}
\right\rbrace,
$$
where $\psi:= (1-\omega^{-1})^{-1}$. The first term in braces is bounded in probability because the $\sigma_i^2$ are finite. For the second term in braces, write $\mathbb{G}_n$ for the empirical cumulative distribution of an i.i.d. sample of size $n$ from the uniform distribution on $[0,1]$ and let
$
\mathbb{G}_n^\prime(u) : = n^{-1} \sum_{i=1}^n \delta_{u_i-u},
$
where $\delta_a$ is Dirac's delta at $a$. Then, writing $\nabla_u$ for the derivative with respect to $u$, we get
\begin{equation} \label{BoundDerivBF2}
\begin{split}
n^{-1} \sum_{i=1}^n \left\lvert \phi^{\prime\prime} \left(\frac{\vartheta_i-\theta}{h}\right)  \right\rvert^\psi
& =  \hspace{.45cm}
{\int}_0^1 \left\lvert \phi^{\prime\prime}\left(\frac{F^{-1}_m(u)-\theta}{h}\right) \right\rvert^\psi \, \mathbb{G}_n^\prime(u) \, du	
\\
& =
-
{\int}_0^1 \nabla_u^1 \left\lvert \phi^{\prime\prime}\left(\frac{F^{-1}_m(u)-\theta}{h}\right) \right\rvert^\psi \, \mathbb{G}_n(u) \, du
\\
& = 
-
{\int}_0^1 \nabla_u^1 \left\lvert \phi^{\prime\prime}\left(\frac{F^{-1}_m(u)-\theta}{h}\right) \right\rvert^\psi \, u \, du 
\\
& \hspace{.45cm}
-
{\int}_0^1 \nabla_u^1 \left\lvert \phi^{\prime\prime}\left(\frac{F^{-1}_m(u)-\theta}{h}\right) \right\rvert^\psi \,  \left( \mathbb{G}_n(u)-u \right) \, du 
\end{split}
\end{equation}		
where we have used integration by parts in the first step and replaced $\mathbb{G}_n(u)$ by $u + (\mathbb{G}_n(u) - u)$ in the second step. We now consider each of the integrals on the right-hand side in turn. First, integrating by parts,
\begin{equation} \label{eq: T1}
-
{\int}_0^1 \nabla_u^1 \left\lvert \phi^{\prime\prime}\left(\frac{F^{-1}_m(u)-\theta}{h}\right) \right\rvert^\psi \, u \, du
=
E\left( \left\lvert \phi^{\prime\prime}\left(\frac{\vartheta_i-\theta}{h}\right) \right\rvert^\psi \right).
\end{equation}
Clearly, this term is bounded uniformly on any finite interval. To evaluate it for large values of $\theta$, observe that
\begin{equation*}
\begin{split}	
\frac{1}{h} \, E\left( \left\lvert  \phi^{\prime\prime}\left(\frac{\vartheta_i-\theta}{h}\right) \right\rvert^\psi \right)
& =
\int_{-\infty}^{+\infty}\frac{1}{h}
\left\lvert  \phi^{\prime\prime}\left(\frac{\vartheta-\theta}{h}\right) \right\rvert^\psi
\, f_m(\vartheta) \, d\vartheta
\\
& =
\int_{\theta - h \log (1+\lvert \theta \rvert)}^{\theta + h \log (1+\lvert \theta \rvert)}\frac{1}{h}
\left\lvert  \phi^{\prime\prime}\left(\frac{\vartheta-\theta}{h}\right) \right\rvert^\psi
\, f_m(\vartheta) \, d\vartheta
\\
& + 
\int_{\log (1+\lvert \theta \rvert)}^{\infty}
\left\lvert  \phi^{\prime\prime}(z) \right\rvert^\psi
\, f_m(\theta + zh) \, dz
\\
& + 
\int_{\log (1+\lvert \theta \rvert)}^{\infty}
\left\lvert  \phi^{\prime\prime}(z) \right\rvert^\psi
\, f_m(\theta - zh) \, dz.
\end{split}
\end{equation*}
Here, 
$$
\int_{\theta - h \log (1+\lvert \theta \rvert)}^{\theta + h \log (1+\lvert \theta \rvert)}\frac{1}{h}
\left\lvert  \phi^{\prime\prime}\left(\frac{\vartheta-\theta}{h}\right) \right\rvert^\psi
\, f_m(\vartheta) \, d\vartheta
\leq O(\log (1+\lvert \theta \rvert)) \, \sup_\theta \lvert f_m(\theta) \rvert
=
O(\log (1+\lvert \theta \rvert)),
$$
because $\sup_\theta \lvert\phi^{\prime\prime}(\theta) \rvert^\psi = O(1)$ and $f_m$ is bounded. Further, because
$$
\textstyle{\int}_x^\infty \lvert \phi^{\prime\prime}(z) \rvert^\psi \, dz
=
O(x^{2\psi-1} \, e^{-\psi \, x^2/2}), \qquad \text{ as } x\rightarrow \infty,
$$
and $f_m(\theta) = O(\lvert \theta \rvert^{-\kappa})$ as $\lvert \theta \rvert \rightarrow \infty$ by Lemma \ref{lemma:helpCDFpreliminaries2}, we have
\begin{equation*}
\begin{split}	
\textstyle{\int}_{\log (1+\lvert \theta \rvert)}^{\infty}
\left\lvert  \phi^{\prime\prime}(z) \right\rvert^\psi
\, f_m(\theta + zh) \, dz
& =
O\left( \log(1+\lvert \theta \rvert)^{2\psi-1} \, e^{-\psi \log(1+\lvert\theta\rvert)^2 /2} \right),
\\
\textstyle{\int}_{\log (1+\lvert \theta \rvert)}^{\infty}
\left\lvert  \phi^{\prime\prime}(z) \right\rvert^\psi
\, f_m(\theta - zh) \, dz
& = 
O\left( \log(1+\lvert \theta \rvert)^{2\psi-1} \, e^{-\psi \log(1+\lvert\theta\rvert)^2 /2} \right).
\end{split}
\end{equation*}
Then, as
$$
e^{-\psi \log(1+\lvert\theta\rvert)^2 /2}  = o(\lvert \theta \rvert^a) \text{ for any } a>0 \text{ as } \lvert \theta \rvert \rightarrow\infty
$$
we may conclude that the term in \eqref{eq: T1} is $O(h \lvert \theta \rvert^{-\kappa} \, \log(1+\lvert\theta \rvert))$ uniformly in $\theta$. Next, for the second term in \eqref{BoundDerivBF2} we use Lemma \ref{lemma:BoundWeightedEmp} to establish that, for any $\epsilon \in (0,1/2]$, we have
\begin{equation*}
\begin{split}	
	&
\left\lvert \textstyle{\int}_0^1 \nabla_u^1 \left\lvert \phi^{\prime\prime}\left(\frac{F^{-1}_m(u)-\theta}{h}\right) \right\rvert^\psi \,   \left( \mathbb{G}_n(u)-u \right) \, du 
\right\rvert
\\ & \leq
o_p(1) \, 
\left\lvert \textstyle{\int}_0^1 \left\lvert \nabla_u^1 \left\lvert \phi^{\prime\prime}\left(\frac{F^{-1}_m(u)-\theta}{h}\right) \right\rvert^\psi \right\rvert \,  \left( u^{1-\epsilon} \, (1-u)^{1-\epsilon}\right) \, du 
\right\rvert
\\
& = 
o_p(1) \,
\left\lvert \textstyle{\int}_{-\infty}^{+\infty} \left\lvert \nabla_u^1 \left\lvert \phi^{\prime\prime}\left(\frac{F^{-1}_m(u)-\theta}{h}\right) \right\rvert^\psi \right\rvert \,  \left( F_m(\vartheta)^{1-\epsilon} \, (1-F_m(\vartheta))^{1-\epsilon}\right) \, d\vartheta
\right\rvert,
\end{split}
\end{equation*}
where the $o_p(1)$ term is independent of $\theta$. The integral term can be bounded in the same way as \eqref{eq: T1}. Hence, 
$$
\left\lvert \textstyle{\int}_0^1 \nabla_u^1 \left\lvert \phi^{\prime\prime}\left(\frac{F^{-1}_m(u)-\theta}{h}\right) \right\rvert^\psi \,   \left( \mathbb{G}_n(u)-u \right) \, du 
\right\rvert
=
o_p(h \lvert \theta\rvert^{(1-\epsilon) \, (1-\kappa)} \, \log (1+\lvert\theta\rvert))
$$ 
uniformly in $\theta$. We therefore have that 
$$
\sup_\theta \lvert \hat{b}_F(\theta) \rvert
\leq
h^{-2} \, O_p(1) \, 
\left\lbrace (O(h \lvert \theta \rvert^{-\kappa} \, \log(1+\lvert\theta \rvert)) +o_p(h \lvert \theta\rvert^{(1-\epsilon) \, (1-\kappa)} \, \log (1+\lvert\theta\rvert))^{\psi^{-1}} \right\rbrace.
$$
For any $ \eta > (\kappa-1)(1-\epsilon)(1-1/\omega) - 1 >0$ it then follows that
\begin{align*}
   \sup_{\theta}  \;
        \left( 1 + \left| \theta \right|^{1+\eta} \right)\, \lvert \hat{b}_F(\theta) \rvert 
  = O_P\left(  h^{-(1+\omega^{-1})}  
             \right) .
\end{align*}
Here, our assumption $\kappa > 1+(1-1/\omega)^{-1}$ guarantees that we can find  $\epsilon>0$
such that $ \eta > (\kappa-1)(1-\epsilon)(1-1/\omega) - 1 >0$ holds. This concludes the proof.
\end{proof}  

\begin{proof}[\bf Proof of Lemma~\ref{lemma:helpCDFpreliminaries5}]
First observe that, for any $\epsilon>0$,
\begin{equation*}
\begin{split}
\sup_\theta E(\lvert b_i(\theta)-E(b_i(\theta)) \rvert^\epsilon)	
& \leq
\sup_\theta\sum_{p=0}^\epsilon \binom{\epsilon}{p}
E(\lvert b_i(\theta) \rvert^p) \, E(\lvert b_i(\theta) \rvert^{\epsilon-p} )
 \leq 2^\epsilon \, \sup_\theta E(\lvert b_i(\theta) \rvert^\epsilon).
\end{split}		
\end{equation*}		
Therefore,
\begin{equation*}
\begin{split}
\sup_\theta E(\lvert b_i(\theta)-E(b_i(\theta)) \rvert^\epsilon)^{\epsilon^{-1}}
& \leq
2 \, \sup_\theta (E(\lvert b_i(\theta) \rvert^\epsilon))^{\epsilon^{-1}}
 \\
 & = 
\sup_\theta
\left( 
\textstyle{\int}_{-\infty}^{\infty} \frac{E(\sigma_i^{2\epsilon} \vert \vartheta_i=\vartheta) \, f_m(\vartheta)}{h^2} \, \left\lvert  \phi^\prime\left(\frac{\vartheta-\theta}{h}\right) \right\rvert^\epsilon \, d\vartheta
\right)^{\epsilon^{-1}}
\\
& \leq
\sup_\vartheta (E(\sigma_i^{2\epsilon} \vert \vartheta_i=\vartheta) \, f_m(\vartheta))^{\epsilon^{-1}}
\,
\frac{\left(\sup_\theta \textstyle{\int}_{-\infty}^{\infty}  \, \left\lvert  \phi^\prime\left(\frac{\vartheta-\theta}{h}\right) \right\rvert^\epsilon \, d\vartheta \right)^{\epsilon^{-1}}}{h^2}
\\
& = 
O(h^{\epsilon^{-1}-2}) ,
\end{split}		
\end{equation*}	
where we have used the definition of $b_i(\theta)$ in the first step, boundedness of the $\sigma_i$ and $f_m$ in the second step, and the fact that
$$
\textstyle{\int}_{-\infty}^{\infty}  \left| \phi^\prime\left(\frac{\vartheta-\theta}{h}\right) \right|^\epsilon  \, d \vartheta = O( h ) ,
$$
independent of $\theta$, in the final step. This completes the proof.
\end{proof}

\section{Least-squares cross validation}

The integrated squared error of 
$$
\check{F}(\theta)=\hat{F}(\theta)-\frac{\hat{b}_F(\theta)}{m}
$$
is 
$$
{\textstyle \int} (\check{F}(\theta)-F(\theta))^2 \, d\theta
=
\frac{{\textstyle \int} \hat{b}_F(\theta)^2 \, d\theta}{m^2}
-
\frac{2{\textstyle \int} (\hat{F}(\theta)-F(\theta)) \, \hat{b}_F(\theta) \, d\theta}{m} + \text{term independent of }h.
$$
Using the definition of $\hat{b}_F$ and expanding the square the first right-hand side term can be written as
\begin{equation*}
\begin{split}	
\frac{\int \hat{b}_F(\theta)^2 \, d\theta}{m^2}
=
\frac{m^{-2}}{n^2} \sum_{i=1}^n \sum_{j=1}^n
\frac{\sigma_i^2 \sigma_j^2}{h^2} \, \frac{1}{4}
\int 
\frac{1}{h}
\phi^\prime\left(\frac{\vartheta_i-\theta}{h} \right) \,
\frac{1}{h} \phi^\prime\left(\frac{\vartheta_j-\theta}{h} \right) \, d\theta,
\end{split}
\end{equation*}
and using properties of the normal distribution 
we calculate
\begin{equation*}
\begin{split}
\int 
\phi^\prime\left(\frac{\vartheta_i-\theta}{h} \right) \,
\phi^\prime\left(\frac{\vartheta_j-\theta}{h} \right) \, d\theta
& = 
\frac{1}{\sqrt{2}h} \phi \left(\frac{\vartheta_i-\vartheta_j}{\sqrt{2}h} \right)
\left(\frac{h^2}{2}-\frac{(\vartheta_i+\vartheta_j)^2}{4}+\vartheta_i \vartheta_j \right).
\end{split}		
\end{equation*}
Next, exploiting that $\phi^\prime(\eta)=-\eta \, \phi(\eta)$ and using well-known results on the truncated normal distribution
\begin{equation*}
\begin{split}	
-\frac{2\int \hat{F}(\theta) \, \hat{b}_F(\theta) \, d\theta}{m}
& = 
\frac{m^{-1}}{n^2}\sum_{i=1}^n \sum_{j=1}^n \frac{\sigma_j^2}{h^2} \int_{\vartheta_i}^{+\infty}  \phi^\prime\left(\frac{\vartheta_j-\theta}{h}\right) \, d\theta
\\
& = 
\frac{m^{-1}}{n^2}\sum_{i=1}^n \sum_{j=1}^n  \frac{\sigma_j^2}{h^2} \int_{\vartheta_i}^{+\infty} 
\left(\frac{\theta-\vartheta_j}{h} \right) \phi\left(\frac{\theta-\vartheta_j}{h}\right) \, d\theta
\\
& = 
\frac{m^{-1}}{n^2}\sum_{i=1}^n \sum_{j=1}^n  \frac{\sigma_j^2}{h} 
\left(\frac{\vartheta_i-\vartheta_j}{h}\right) \, \phi\left(\frac{\vartheta_i-\vartheta_j}{h}\right)
\\
& = 
\frac{m^{-1}}{n^2}\sum_{i=1}^n \sum_{j\neq i} 
\frac{\sigma_i^2}{h} \phi^\prime\left(\frac{\vartheta_i-\vartheta_j}{h}\right) .
\end{split}
\end{equation*}
Omitting terms for which $j=i$ in the last expression is justified by the fact that $\phi^\prime(0)=0$.
Finally, for the last term, integrating by parts shows that
\begin{equation*}
\begin{split}	
\frac{2\int F(\theta) \, \hat{b}_F(\theta) \, d\theta}{m}
& = 
-\frac{m^{-1}}{n} \sum_{i=1}^n \frac{\sigma_i^2}{h} \,
\int \phi\left(\frac{\vartheta_i-\theta}{h}\right) \, f(\theta) \, d\theta .
\end{split}
\end{equation*}
The integral in the right-hand side expression represents an expectation taken with respect to $f$. A leave-one-out estimator of the entire term is
$$
-\frac{m^{-1}}{n(n-1)} \sum_{i=1}^n \sum_{j\neq i} \frac{\sigma_i^2}{h} \phi\left(\frac{\vartheta_i-\vartheta_j}{h}\right) .
$$
Combining results and multiplying the entire expression through with $n^2m^2$ yields the cross-validation objective function stated in the main text.

\end{document}